\documentclass[journal,twocolumn]{IEEEtran}
\usepackage{amsmath,amssymb}
\usepackage{mathrsfs}
\usepackage{latexsym}
\usepackage{psfrag}
\usepackage{graphicx,subfigure}
\usepackage{graphicx}
\usepackage{flafter}
\usepackage{amsmath}
\usepackage{amsfonts}
\usepackage{amssymb}
\usepackage{latexsym}
\usepackage{psfrag}
\usepackage{graphicx,subfigure}
\usepackage{hyperref}
\usepackage{url}
\usepackage{epstopdf}
\usepackage{color}
\usepackage{cite}
\usepackage{mathtools}
\usepackage{mathrsfs}
\usepackage{pifont}
\allowdisplaybreaks

\newtheorem{theorem}{Theorem}
\newtheorem{definition}{Definition}
\newtheorem{corollary}{Corollary}
\newtheorem{example}{Example}
\newtheorem{assumption}{Assumption}

\newtheorem{lemma}{Lemma}
\newtheorem{remark}{Remark}
\newtheorem{proposition}{Proposition}
\newtheorem{problem}{Problem}

\newcommand{\bel}[1]{\begin{eqnarray}\label{#1}}

\newcommand{\be}{\begin{eqnarray}}
\newcommand{\ee}{\end{eqnarray}}

\newcommand{\bac}{\begin{array}{cccccccccccccccccccccccc}}
\newcommand{\ea}{\end{array}}
\newcommand{\bal}{\begin{array}{llllllllllllllllllllllll}}
\newcommand{\bpr}{\begin{proposition}}
\newcommand{\epr}{\end{proposition}}

\newcommand{\bt}{\begin{theorem}}
\newcommand{\et}{\end{theorem}}
\newcommand{\bd}{\begin{definition}\rm}
\newcommand{\ed}{\end{definition}}
\newcommand{\bc}{\begin{corollary}}
\newcommand{\ec}{\end{corollary}}
\newcommand{\bex}{\begin{example}\rm}
\newcommand{\eex}{\end{example}}
\newcommand{\bl}{\begin{lemma}\rm}
\newcommand{\el}{\end{lemma}}
\newcommand{\br}{\begin{remark}\rm}
\newcommand{\er}{\end{remark}}
\newcommand{\bas}{\begin{assumption}\rm}
\newcommand{\eas}{\end{assumption}}
\newcommand{\bp}{\begin{proof}}
\newcommand{\ep}{\end{proof}}
\def\exp{\mbox{exp}}
\def\max{\mbox{max}}

 
\begin{document}
\title{Real-time Estimation of DoS Duration and Frequency for Security Control}
\author{Yifan Sun, \IEEEmembership{Graduate Student Member,~IEEE}, Jianquan Lu, \IEEEmembership{Senior Member,~IEEE},\\ Daniel W. C. Ho, \IEEEmembership{Life Fellow,~IEEE} and Lulu Li
\thanks{This work was partially supported by the Natural Science Foundation of Jiangsu Province under Grant No. BK20240009, and the National Natural Science Foundation of China under Grant No. 62373105, and Jiangsu Provincial Scientific Research Center of Applied Mathematics under Grant BK20233002, and the Research Grants Council of the Hong Kong Special Administrative Region, China (CityU 11203521, 11213023).
\textit{(Corresponding Author: Jianquan Lu.)}}
\thanks{Yifan Sun is with the School of Mathematics, Southeast University, Nanjing 211189, Jiangsu, China (email: yfsun\_xdu@163.com).}
\thanks{Jianquan Lu is with the School of Mathematics, Southeast University, Nanjing 210096, Jiangsu, China (email: jqluma@seu.edu.cn).}
\thanks{Daniel W. C. Ho is with the Department of Mathematics, City University of Hong Kong, Hong Kong, China (email: madaniel@cityu.edu.hk).}
\thanks{Lulu Li is with the School of Mathematics, Hefei University of Technology, Hefei 230601, Anhui, China (email: lululima@hfut.edu.cn).}
}
\maketitle
\begin{abstract}


In this paper, we develop a new \textit{denial-of-service (DoS) estimator},
enabling defenders to identify duration and frequency parameters of any DoS attacker, except for three \textit{edge cases}, exclusively using real-time data.
The key advantage of the estimator lies in its capability to facilitate security control in a wide range of practical scenarios, even when the attacker's information is previously unknown. We demonstrate the advantage and application of our new estimator in the context of two classical control scenarios, namely consensus of multi-agent systems and impulsive stabilization of nonlinear systems, for illustration.

\end{abstract}
\section{Introduction}
\textit{Denial-of-service} (DoS) attack is a form of cyber-attack that makes a machine or network resource unavailable to its intended users. In a networked control system, it usually affects the timeliness of communication by blocking the channel and inducing unexpected transmission failures, and thus invalidates most classical control methods and causes the system to run into a hazardous state \cite{Persis_DoS_tac2015,DoS_in_sensor_networks_Computer_2002,YANG2020109116}. Therefore, it is necessary to take security into account when designing controllers, or otherwise, the system will remain vulnerable to DoS attacks \cite{DoS_in_sensor_networks_Computer_2002}.

In order to model the security control problem, a comprehensive characterization of DoS attacks is essential. It is inadequate to assume a specific probabilistic distribution \cite{CHENG2024111470,AMIN_auto_2013} or a certain action mode (e.g., periodicity \cite{WZheng_periodicDoS_tac2020}) of the attack because a malicious attacker does not have to follow these assumptions and can exhibit quite irregular behavior. For better universality, De Persis and Tesi \cite{Persis_DoS_tac2015} proposed a DoS model that does not assume about the attacker's underlying strategy. This model is capable of capturing various types of DoS patterns, and therefore, has been widely used in follow-up research on anti-DoS control \cite{DePersis_DoS_nonlinear_scl2016,SFeng_DoS_SCL2017,HIshii_dos_quantization_tac2020,YE2023111170,PTesi_ETCDoS_TCNS2017}.

Under the framework \cite{Persis_DoS_tac2015}, \textit{DoS duration} and \textit{DoS frequency} are two core metrics characterizing DoS attacks. The former metric evaluates the percentage of time when a DoS attack is active, while the latter characterizes how frequently the attack is launched (both in an average sense \textit{over the entire timeline}).
The advantage of the framework \cite{Persis_DoS_tac2015} is to allow controller configurations based only on the duration and frequency parameters regardless of any other detailed attack information, thus leading to very concise control strategies\cite{WXu_distributedDoS_tcyb2020,HIshii_dos_quantization_tac2020,RKato_linearization_tac2022,SFeng_DoS_Quantization_2022_auto,DZhang_outputregulation_tac2023}. For example, if a duration- \& frequency-based controller setting is effective for the periodic DoS attack over $[0,1)$, $[2,3)$, $[4,5)$, $\cdots$, then it is also effective for an aperiodic attack over $[1,2)$, $[3,4-\ln{2})$, $[5,6+\ln{\frac{3}{2}})$, $\cdots$, $[2n+1,2n+2+ (-1)^n \ln{\frac{n+1}{n}})$, $\cdots$. More precisely, both these attacks have the same duration and frequency parameters, that is, affecting around $50\%$ of the time and being launched every two seconds.

Unfortunately, the above advantage of the framework \cite{Persis_DoS_tac2015} can hardly be implemented in most practical applications. It is because the framework relies on \textit{directly assumed} values of the duration and frequency parameters whose correctness cannot be guaranteed in real cases where the defender usually has little adversary information before the encounter. To validate the framework \cite{Persis_DoS_tac2015} in these cases, the defender should estimate the parameters of interest based only on real-time DoS data, which can be collected from \textit{past} actions of the attacker.


Motivated by the discussion above, the focus of this paper lies in achieving a dependable estimation of DoS duration and frequency parameters utilizing solely real-time attack data.
As the theoretical support of our estimation, we reveal a fundamental vulnerability inherent in DoS attackers: they \textit{unavoidably disclose information about the duration and frequency through their actions}. This inherent weakness inspires the development of our novel \textit{DoS estimator}, which can identify the desired parameters by examining the past attack durations and the historical count of DoS launches.


Despite some involved analyses ensuring theoretical validity, the core idea of our estimation is remarkably intuitive. As defenders, we consider slightly higher DoS duration and frequency than the worst previous scenario, anticipating that such estimation will achieve reliability \textit{within a finite time}. With our \textit{DoS estimator} designed following this idea, a satisfactory performance can be achieved against virtually all types of DoS attackers except for three extreme \textit{edge cases}. These atypical scenarios correspond to three categories of extremely powerful attackers, which are rarely encountered in real network environments. To further elaborate on this point, we will provide a new and concise definition in Section \ref{sec:DoS_characterization}, justifying it based on the physical meanings of the three cases.


To the authors' best knowledge, few existing works have considered the present topic despite its undoubted importance. The reference closest to this paper might be \cite{SFeng_DoS_Quantization_2022_auto}, in which available transmission bits are adaptively allocated based on real-time DoS strength. Yet, how to estimate the attacker's duration \& frequency parameters remains unresolved in that article. The work merely provided a precise solution to the allocation problem of channel resources, but the design of the sampling period there still relies on presupposed attacker information. Another relevant reference is \cite{Active_security_tac2021}, in which the DoS duration or frequency is not required to be known by the defender. However, the controller design is based on the maximum duration of the DoS attack, which still needs to be assumed before the design time and, therefore, does not change the essence of the problem. In \cite{real-time_detection_and_estimation_DoS}, a real-time scheme is proposed to detect and estimate DoS attacks that lead to a \textit{differentiable} transmission delay. Most real-world cases cannot be covered because DoS-induced delays are usually discontinuous. Li \textit{et al.} proposed another viewpoint in \cite{LShi_jamming_game_tac2015}, formulating a game-theoretic framework to describe the interaction between attackers and defenders. Although the reliance on duration \& frequency parameters can be avoided, such a framework requires that both sides know the energy constraints of their adversaries before the game begins.

The main contributions of this paper are in three different aspects.
Firstly, we develop a new \textit{DoS estimator} capable of addressing any DoS attacker apart from three \textit{edge cases}, identifying the duration and frequency parameters for defensive purposes solely based on real-time data.
This novel development has the potential to significantly broaden the applicability of numerous existing anti-DoS control strategies, specifically extending their reach from managing rare, previously identified attackers to addressing all but a few extreme ones.


The second goal of this paper is to provide clear guidance on exploiting the \textit{DoS estimator} for security controller design. To achieve this, we examine two specific control scenarios, namely the consensus of multi-agent systems (MASs) and impulsive stabilization of nonlinear systems, for illustration. Utilizing the new estimator, we present concrete design schemes that exhibit satisfactory control performances, even against previously unknown attackers, in both cases.
Based on our estimator, the design reflects the feasibility of expanding the application scope of previous control strategies.

Finally, the design presented in this paper offers several parametric tuning options to balance security against cost-efficiency.
These options provide flexibility for a conservative defender to allocate more resources for faster estimation (resulting in enhanced security), or for an aggressive defender to prioritize energy-saving over the speed of estimation. The choice between these options depends on how skeptical the defender is about the attacker's behavior.
The trade-off therefore reveals an interesting interplay between the offensive and defensive sides during the estimation process, and it offers a possibility for controller configuration to meet diverse requirements.

Standard notations will be used throughout this paper.
We use $\mathbb{N}$ (resp., $\mathbb{N}_+$) to represent the set of all non-negative integers (resp., positive integers), and use $\mathbb{R}^n$ (resp., $\mathbb{R}$) to represent the $n$-dimensional (resp., $1$-dimensional) Euclidean vector space.
We use $\mathbb{R}_{\geqslant 0}$ to represent the set of non-negative real numbers. For a set $\Omega$, we use $\operatorname{card}(\Omega)$ to represent its cardinality. In addition, if $\Omega \subseteq \mathbb{R}$, then we use $\lvert \Omega \rvert$ to represent its Lebesgue measure \cite{book_real_analysis}, use $\inf \Omega$ to represent its infimum, and use $\sup \Omega$ to represent its supremum \cite{Rudin_MA_mcgraw-hill1976}.
We use $\limsup_{n \rightarrow \infty}{a_n} \coloneqq \lim_{n \rightarrow \infty} \left( \sup_{m \geqslant n} a_m\right)$ to represent the limit superior \cite{Rudin_MA_mcgraw-hill1976} of a sequence $\{a_n\}_{n \in \mathbb{N}_+}$.
For a function $f(t)$, we use $\limsup_{t \rightarrow +\infty}{f(t)} \coloneqq \lim_{t \rightarrow +\infty} \left( \sup_{\tau \geqslant t} f(\tau)\right)$ to represent its limit superior, use $\dot{f}(t)$ to represent its derivative with respect to $t$, and use $D^+ f (t)$ to represent its upper-right Dini derivative \cite{TYang_impulsive-control_TAC1999}. We say that $f: \mathbb{R}_{\geqslant 0} \mapsto \mathbb{R}_{\geqslant 0} \in \mathcal{K}$, if $f$ is continuous, $f(0) = 0$ and $f(s)$ is strictly increasing in $s$. For a vector $x$, we use $x^T$ to represent its transposition and use $\lvert x \rvert$ to represent its Euclidean norm. We use $\mathbf{1}_N$ to represent an $N$-dimensional vector $(1,1,\cdots,1)^T$.


\section{DoS estimation} \label{sec:DoS_characterization}
\subsection{Problem formulation}
Since the main purpose of this paper is to give reliable estimation of DoS duration and frequency parameters, firstly, we need to review their definitions proposed by \cite{Persis_DoS_tac2015}.

Let $\{h_n\}_{n \in \mathbb{N}_+}$ represent the sequence of instants when the DoS attack switches from \textit{off} to \textit{on}, where $h_1 \geqslant 0$. Let
\begin{align} \label{align:H_n}
    H_n \coloneqq \{h_n\} \cup [h_n, h_n + \tau_n)
\end{align}
represent the $n$-th DoS time interval of a length $\tau_n$, over which the attack is \textit{on}, and all transmission attempts are denied. Notably, the $n$-th interval $H_n$ reduces to a single point $\{h_n\}$ if $\tau_n = 0$. Without loss of generality, we assume that $h_{n+1} > h_n + \tau_n$ for any $n \in \mathbb{N}_+$ so that $H_i \cap H_j = \emptyset$ for any $i \neq j$. We use the notation
\begin{align} \label{align:Xi}
    \Xi (\tau,s) \coloneqq \bigcup_{n\in \mathbb{N}_+} H_n \cap [\tau,s]
\end{align}
to denote the set of instants subject to DoS attack on $[\tau,s]$ ($\tau \leqslant s$), and use
\begin{align} \label{align:Theta}
    \Theta (\tau,s) \coloneqq [\tau,s] \setminus \Xi (\tau,s)
\end{align}
to denote the complement of $\Xi (\tau,s)$, over which the attack is \textit{off} and transmission is allowed. In this paper, we use
$\xi \coloneqq \{H_n\}_{n \in \mathbb{N}_+}$
to denote the sequence of all the DoS time intervals, and we call $\xi$ a \textit{DoS sequence} for brevity. For notation consistency, we specify that $h_{m+1} = h_{m+1} + \tau_{m+1} = +\infty$ if there are only $m$ elements in $\xi$, $m \in \mathbb{N}$. Finally, we denote by
\begin{align} \label{align:n_xi}
    n_{\xi} (\tau,s) \coloneqq \operatorname{card}\left( \{h_n\}_{n \in \mathbb{N}_+} \cap [\tau,s] \right)
\end{align}
the number of DoS {\textit{off}-to-\textit{on}} transitions on $[\tau,s]$. The following example is provided to help understand the notations in (\ref{align:H_n})--(\ref{align:n_xi}).

\begin{example}
    Let us consider $h_n = n$ and $\tau_n = 0.5$ for any $n \in \mathbb{N}_+$. Then, $H_n = \left[n, n+ 0.5\right)$ is the $n$-th DoS time interval, and $\xi$ is an infinite sequence of time intervals, namely
    $$\xi = \left\{ \left[1, 1.5\right), \left[2, 2.5\right), \left[3, 3.5\right), \cdots \right\}.$$
    In addition, if $\tau = 2.3$ and $s = 5.6$, then it follows from (\ref{align:Xi}) that
    \begin{align*}
        \Xi(\tau,s) &= \big(\left[1, 1.5\right) \cup \left[2, 2.5\right) \cup \left[3, 3.5\right) \cup \cdots \big) \cap [2.3, 5.6] \\
        &=[2.3, 2.5) \cup [3, 3.5) \cup [4,4.5) \cup [5,5.5),
    \end{align*}
    from (\ref{align:Theta}) that
    \begin{align*}
        \Theta(\tau,s) &= [2.3, 5.6] \setminus \Xi(\tau,s) \\
        &= [2.5, 3) \cup [3.5, 4) \cup [4.5,5) \cup [5.5, 5.6),
    \end{align*}
    and from (\ref{align:n_xi}) that
    $$n_{\xi} (\tau,s) = \operatorname{card} \{3,4,5\} = 3.$$
\end{example}

With the preparations above, the duration and frequency parameters of a DoS sequence $\xi$ can be defined as follows.
\begin{definition} \label{def:duration}
    \textit{(\cite{Persis_DoS_tac2015} DoS duration)} For a DoS sequence $\xi$ and a constant $B_d \in [0, 1]$, if there exists a constant $0 < \kappa (B_d) < +\infty$ such that
        \begin{align*}
            \lvert \Xi (0,t) \rvert \leqslant \kappa + B_d \, t, \quad \forall \, t \geqslant 0,
        \end{align*}
    then we say that $B_d$ is a \textit{duration-bound} of $\xi$. \hfill $\square$
\end{definition}


\begin{definition} \label{def:frequency}
    \textit{(\cite{Persis_DoS_tac2015} DoS frequency)} For a DoS sequence $\xi$ and a constant $B_f \in [0, +\infty)$, if there exists an integer $0 < \Lambda (B_f) < +\infty$ such that
        \begin{align*}
            n_{\xi} (0,t) \leqslant \Lambda + B_f \, t, \quad \forall \, t \geqslant 0,
        \end{align*}
    then we say that $B_f$ is a \textit{frequency-bound} of $\xi$. If, otherwise, no such constant $B_f$ exists, then the \textit{frequency-bound} of $\xi$ is defined to be $B_f = +\infty$. \hfill $\square$
\end{definition}

The \textit{duration-bound} characterizes the ratio of time over which the DoS attack is active, and the \textit{frequency-bound} characterizes how frequently the attack is launched. Generally, the values of these two metrics do not have to affect each other, and therefore, we can deal with $B_d$ and $B_f$, respectively, in order to simplify analysis.

According to \textit{Definitions \ref{def:duration}} and \textit{\ref{def:frequency}}, $B_d$ and $B_f$ are related to the attacker's action on $[0,+\infty)$. Their values, however, are directly assumed by the defender when $t = 0$ in most previous works \cite{Persis_DoS_tac2015,DePersis_DoS_nonlinear_scl2016,SFeng_MAS_TCNS2022,PTesi_ETCDoS_TCNS2017,SFeng_DoS_SCL2017,WXu_distributedDoS_tcyb2020,HIshii_dos_quantization_tac2020,RKato_linearization_tac2022,SFeng_DoS_Quantization_2022_auto,DZhang_outputregulation_tac2023}.
Such presupposed values are reliable only if the defender has accurate \textit{future} information about the attacker, but not for most real cases where nothing about $[t,+\infty)$ is known at time $t$. In order to manage the latter cases, $B_d$ and $B_f$ should be estimated in real-time based only on the attacker's \textit{past} actions over $[0,t]$.

Inspired by the discussion above, the first task of a defender is to give a real-time estimation of $B_d$ and $B_f$ merely utilizing past knowledge about the attacker. For the defense effectiveness, the estimation is required to become reliable within finite time. If we denote by $\hat{B}_d (t)$ and $\hat{B}_f (t)$, respectively, the estimated value of $B_d$ and $B_f$ at time $t$, and by
$$\mathscr{K}_a (0,t) \coloneqq \Big\{ \big(s, \lvert \Xi (0,s) \rvert, n_{\xi} (0,s) \big) \Big| s \in [0,t] \Big\} $$
the available knowledge about the attacker on $[0,t]$, then, the task above can be written more formally as follows.
\begin{problem} \label{prob:estimator}
    \textit{(Real-time DoS estimation)} Design a \textit{DoS estimator} as
    \begin{align*}
        \hat{B}_d (t) = \hat{B}_d \left( \mathscr{K}_a (0,t) \right),\quad \hat{B}_f (t) = \hat{B}_f \left( \mathscr{K}_a (0,t) \right),
    \end{align*}
    such that there exists a finite $T$ satisfying
    \begin{itemize}
        \item[a)] $\hat{B}_d (t)$ is a \textit{duration-bound} of $\xi$ for any $t\geqslant T$.
        \item[b)] $\hat{B}_f (t)$ is a \textit{frequency-bound} of $\xi$ for any $t\geqslant T$.
    \end{itemize}
\end{problem}

\begin{remark}
    As discussed in \cite{PTesi_ETCDoS_TCNS2017}, the exponential convergence rate of a control system is typically affected by the \textit{duration-} \& \textit{frequency-bounds}, while the $\mathcal{L}_\infty$-gain (from noise to the output) also depends on $\kappa$ and $\Lambda$ in \textit{Definitions \ref{def:duration} }and\textit{ \ref{def:frequency}}. Our focus here is on the asymptotic (or exponential) convergence of noiseless systems, so estimating only $B_d$ and $B_f$ suffices. In the context of $\mathcal{L}_\infty$-gain problems, estimating $\kappa$ and $\Lambda$ would also be necessary, which is, however, out of the scope of this paper and is not investigated.
\end{remark}

\subsection{Solution to \textit{Problem \ref{prob:estimator}}: The \textit{DoS estimator}}
This subsection aims at solving \textit{Problem \ref{prob:estimator}}. To this end, we first need to find connections between the attacker's past knowledge $\mathscr{K}_a (0,t)$ and its \textit{duration-bound} $B_d$ \& \textit{frequency-bound} $B_f$.
This is not easy since $B_d$ and $B_f$ are not necessarily unique. In fact, they merely provide a loose estimation of how strong the attacker can be. Defensive strategies based on such a loose estimation could be inefficient because the defender must prepare for the worst scenario, even if it never actually happens.

To make the defense more efficient, we have to investigate the sharp bounds of $B_d$ and $B_f$.
For this purpose, let us consider the set composed of all \textit{duration-bounds} of $\xi$, namely,
\begin{align*}
\mathcal{D} (\xi) \coloneqq \left\{B_d \in [0,1] \mid B_d\text{ is a \textit{duration-bound} of $\xi$} \right\}.
\end{align*}
By \textit{Definition \ref{def:duration}}, $\mathcal{D} (\xi)$ is nonempty because we always have $1 \in \mathcal{D} (\xi)$. Hence, the infimum of $\mathcal{D} (\xi)$ is well-defined and satisfies $\inf \mathcal{D} (\xi) \leqslant 1$.
Similarly, for the \textit{frequency-bound}, we can define
\begin{align*}
\mathcal{F} (\xi) \coloneqq \left\{B_f \in [0,+\infty] \mid B_f\text{ is a \textit{frequency-bound} of $\xi$} \right\},
\end{align*}
and we have 
$\inf \mathcal{F} (\xi) \leqslant +\infty$.
Although $B_d$ (resp., $B_f$) is not unique, there must be only one infimum of $\mathcal{D} (\xi)$ (resp., $\mathcal{F} (\xi)$). In addition, we have the following proposition which demonstrates that $\inf \mathcal{D} (\xi)$ and $\inf \mathcal{F} (\xi)$ are exactly the desired sharp bounds of $B_d$ and $B_f$, respectively.
\begin{proposition} \label{prop:sharp_bounds}
Given any DoS sequence $\xi$, we have
\begin{align} \label{align:Dxi}
    {\mathcal{D}^{\circ} (\xi) = (\inf \mathcal{D} (\xi),1)}
\end{align}
and
\begin{align} \label{align:Fxi}
    {\mathcal{F}^{\circ} (\xi) = (\inf \mathcal{F} (\xi),+\infty),}
\end{align}
where $\Omega^{\circ}$ represents the interior of a set $\Omega$.
\end{proposition}
\begin{proof}
    We only prove (\ref{align:Dxi}) since the proof of (\ref{align:Fxi}) is analogous. By definition, there must hold $\mathcal{D} (\xi) \subseteq [\inf \mathcal{D} (\xi),1]$. If we assume by contradiction that (\ref{align:Dxi}) is false, there would exist a $B_d^\ast \in (\inf \mathcal{D} (\xi),1)$ that is not a \textit{duration-bound} of $\xi$. This, however, implies that
    $$\mathcal{D} (\xi) \subseteq (B_d^\ast,1],$$
    which contradicts against the definition of infimum. Therefore, our initial assumption must be false, meaning that (\ref{align:Dxi}) is true. \hfill $\square$
\end{proof}



With the preparations above, finding the connections between $\mathscr{K}_a (0,t)$ and $B_d$ \& $B_f$ is equivalent to finding the ones between $\mathscr{K}_a (0,t)$ and $\inf \mathcal{D} (\xi)$ \& $\inf \mathcal{F} (\xi)$. The following result can establish such connections.


\begin{lemma} \label{lemma:function_limit}
    Given any DoS sequence $\xi$, we have
    \begin{align} \label{align:duration_function_limit}
        \limsup_{t \rightarrow + \infty}{\frac{\lvert \Xi (0,t) \rvert}{t}} = \inf \mathcal{D} (\xi)
    \end{align}
    and
    \begin{align} \label{align:frequency_function_limit}
        \limsup_{t \rightarrow + \infty}{\frac{n_{\xi} (0,t)}{t}} = \inf \mathcal{F} (\xi).
    \end{align}
\end{lemma}
\begin{proof}
    See Appendix \ref{appA}. \hfill $\square$
\end{proof}

\begin{remark}
    As revealed by \textit{Lemma \ref{lemma:function_limit}}, the values of $\inf \mathcal{D}(\xi)$ \& $\inf \mathcal{F} (\xi)$ have explicit connections with certain behaviors of the attacker. More specifically, in the sense of limit superior, the ratio of time under attack converges to $\inf \mathcal{D} (\xi)$, and the frequency of launches converges to $\inf \mathcal{F} (\xi)$. The results in (\ref{align:duration_function_limit}) and (\ref{align:frequency_function_limit}) provide the possibility of estimating the \textit{duration} \& \textit{frequency bounds}, respectively, based on real-time DoS data $\lvert \Xi (0,t) \rvert$ and $n_{\xi} (0,t)$, and thus are crucial for addressing \textit{Problem \ref{prob:estimator}}.
\end{remark}


Based on \textit{Lemma \ref{lemma:function_limit}}, we point out three \textit{edge cases} of $\xi$ which are beyond our ability before providing a solution to \textit{Problem \ref{prob:estimator}}. These \textit{edge cases} correspond to three types of extremely powerful DoS attackers that rarely exist in real networks. More formally, we have the following definition of these cases.
\begin{definition} \label{def:edge_cases}
    In this paper, a DoS sequence $\xi$ is called an \textit{edge case}, if it satisfies any of the following.
    \begin{itemize}
        \item[i)] $\inf \mathcal{D} (\xi) = 1$.
        \item[ii)] $\inf \mathcal{F} (\xi) = +\infty$.
        \item[iii)] For any $\Gamma \in (0,+\infty)$, there exists an $n \in \mathbb{N}_+$ such that $\tau_{n} > \Gamma.$ Here, $\tau_n$ is the length of the $n$-th DoS interval as defined in (\ref{align:H_n}).
    \hfill $\square$
    \end{itemize}
\end{definition}

    In this paper, the defender becomes powerless under the \textit{edge cases} in \textit{Definition \ref{def:edge_cases}}. Nevertheless, this is acceptable
    because, from an attacker's perspective, being an \textit{edge case} means being extremely powerful and quite demanding. We can use \textit{Lemma \ref{lemma:function_limit}} to elaborate on this point. For i) in the definition, it corresponds to the total duration of attacks and means that $\xi$ should attack and corrupt almost the entire interval $[0,t]$ even when $t \rightarrow +\infty$ according to (\ref{align:duration_function_limit}). For ii), it corresponds to the frequency of attacks, meaning that the attack should be launched at an infinitely high rate, according to (\ref{align:frequency_function_limit}). For iii), it means that the duration $\tau_n$ of a single attack cannot be upper-bounded by any finite constant, indicating the unlimited storage capacity of the attacker.
    Since it is rare to encounter such strong opponents in practice, the inability to deal with them has little impact on the universality of our solution. Hence, the following assumption is imposed on $\xi$ hereafter.

    \begin{assumption} \label{assumption:not_edge_case}
        The DoS sequence $\xi$ is not an \textit{edge case}.
    \end{assumption}
    
Now, we proceed to address \textit{Problem \ref{prob:estimator}}.
Let us design the \textit{DoS estimator} as
    \begin{align} \label{align:hatT_A(n)}
        \hat{B}_d (t) \coloneqq 
        \left\{
        \begin{array}{lcl}
            \epsilon_0,\qquad t \in [0,h_\ell+\tau_\ell), \\
            \max_{\ell\leqslant i \leqslant n} \left\{ \epsilon_0,\theta B_d (i) + (1-\theta) \right\},\\
            \qquad \quad t \in [h_n + \tau_n, h_{n+1} + \tau_{n+1})
        \end{array}
        \right.
    \end{align}
    and
    \begin{align} \label{align:hattau_A(n)}
        \hat{B}_f (t) \coloneqq
        \left\{
        \begin{array}{ll}
            \epsilon_0, & t \in [0,h_\ell), \\
            \max_{\ell \leqslant i \leqslant n} \left\{\epsilon_0, \frac{B_f (i)}{\theta} \right\}, & t \in [h_n, h_{n+1}).
        \end{array}
        \right.
    \end{align}
    Here, $2 \leqslant \ell \in \mathbb{N}_+$, $0< \epsilon_0 < 1$ and $0< \theta < 1$ are tunable parameters. They can be flexibly chosen and are independent of $\xi$. In addition, $B_d (i) \coloneqq \frac{\lvert \Xi(0,h_i + \tau_i) \rvert}{h_i + \tau_i}$, and $B_f (i) \coloneqq \frac{n_\xi (0,h_i)}{h_i} = \frac{i}{h_i}$, where the notations $h_n$, $\Xi (\tau,s)$ and $n_{\xi} (\tau,s)$ are defined, respectively, in (\ref{align:H_n}), (\ref{align:Xi}) and (\ref{align:n_xi}). For any $t\geqslant 0$, it is clear that $\hat{B}_d (t)$ and $\hat{B}_f (t)$ generated by (\ref{align:hatT_A(n)})--(\ref{align:hattau_A(n)}) only rely on past information over $[0,t]$.

    The following result shows that \textit{DoS estimator} (\ref{align:hatT_A(n)})--(\ref{align:hattau_A(n)}) is a satisfactory solution to \textit{Problem \ref{prob:estimator}}.

\begin{theorem} \label{theorem:bound_reconstruction}
    Let the estimation $\hat{B}_d (t)$ and $\hat{B}_f (t)$ be generated by \textit{DoS estimator} (\ref{align:hatT_A(n)})--(\ref{align:hattau_A(n)}). Then under \textit{Assumption \ref{assumption:not_edge_case}}, there must exist a finite $T \geqslant 0$, such that
    $\hat{B}_d (t) \in \mathcal{D}(\xi)$ and $\hat{B}_f (t) \in \mathcal{F}(\xi)$ hold  for any $t\geqslant T$.
\end{theorem}

\begin{proof}
    See Appendix \ref{appB}. \hfill $\square$
\end{proof}

\begin{remark}
    Unlike many conventional approaches \cite{Persis_DoS_tac2015,DePersis_DoS_nonlinear_scl2016,SFeng_DoS_SCL2017,HIshii_dos_quantization_tac2020,SFeng_MAS_TCNS2022,WXu_distributedDoS_tcyb2020} that rely on pre-assumed values of $B_d$ and $B_f$, \textit{DoS estimator} (\ref{align:hatT_A(n)})--(\ref{align:hattau_A(n)}) allows for controller design based on their real-time estimations, namely $\hat{B}_d (t)$ and $\hat{B}_f (t)$. According to \textit{Theorem \ref{theorem:bound_reconstruction}}, such an estimator-based design will become reliable after a finite period of time as long as \textit{edge cases} do not occur. We will further elaborate on this point in Sections \ref{sec:consensus} and \ref{sec:impulsive_stabilization} by examining two specific control scenarios and providing concrete design schemes for them based on the \textit{DoS estimator}.
\end{remark}

\begin{remark} \label{remark:theta}
    The choice of $\theta$ ($0 < \theta < 1$) is crucial and clearly reflects our basic idea. The left-hand side inequality $\theta > 0$ means that we surrender under \textit{edge cases} i) and ii) (because $\theta> 0$ implies that $\hat{B}_d (t) < 1$ and $\hat{B}_f (t) < +\infty$, which is a false estimation in the first two \textit{edge cases}). Meanwhile, the right-hand side $\theta<1$ means that we keep being skeptical about the observed data $B_d (i)$ \& $B_f (i)$ and consider scenarios slightly worse than that. It is important that $\theta$ should be strictly less than $1$, or otherwise, the wrong estimation might be obtained. For example, for $\xi = \{[2n+1,2n+2)\}_{n \in \mathbb{N}_+}$, we have $\inf \mathcal{D} (\xi) = 0.5$, but $B_d (i) = \frac{\lvert \Xi(0,h_i + \tau_i) \rvert}{h_i + \tau_i} = \frac{i}{2i + 2} < 0.5$ for any $i \in \mathbb{N}_+$. In such a case, if we let $\theta = 1$, then the estimator $\hat{B}_d (t)$ fails to identify the \textit{duration-bound} of $\xi$ even when $t \rightarrow +\infty$. We will further illustrate this counterexample in \textit{Example \ref{ex3}}.
\end{remark}

\begin{remark}
    The effectiveness of estimator (\ref{align:hatT_A(n)})--(\ref{align:hattau_A(n)}) relies on the convergence of the two infinite data sequences $B_d (i)$ and $B_f (i)$, as demonstrated in the proof of \textit{Theorem \ref{theorem:bound_reconstruction}}. However, in the \textit{edge case} iii) where a single attack's impact can be unbounded, the attacker may achieve its objective with only a finite number of attacks. This scenario renders our estimator ineffective because the two data sequences become finite and no longer exhibit any convergence behavior.
\end{remark}

\begin{remark} \label{remark:parameter_tuning}
    In the DoS estimator (\ref{align:hatT_A(n)}) and (\ref{align:hattau_A(n)}),
    selecting a larger value of $\ell$ means letting the estimation process in the second line begin later, and selecting a larger $\theta$ means being less skeptical about the attacker's past actions (because that makes the generated bounds closer to the observed $B_d (i)$ and $B_f (i)$, respectively). Intuitively, setting them larger will eventually lead to tighter \textit{duration} \& \textit{frequency bounds} (and consequently, fewer resources are required for defense), but the price is that more time is required to generate a reliable estimation (namely, $T$ in \textit{Theorem \ref{theorem:bound_reconstruction}} becomes larger). As the result, we might have to sacrifice the transient performance as the trade-off, which will be further illustrated in \textit{Example \ref{ex2}}.
\end{remark}

\begin{remark} \label{remark:epsilon0}
    The parameter $\epsilon_0$ serves as a default value of \textit{DoS estimator} (\ref{align:hatT_A(n)})--(\ref{align:hattau_A(n)}). Choosing a larger $\epsilon_0$ helps improve the transient performance of control systems, but it requires the defender to invest more defensive resources as the trade-off. This point will be illustrated in \textit{Example \ref{ex_epsilon0}}.
\end{remark}

\subsection{Time Required for A Reliable Estimation}
According to \textit{Theorem \ref{theorem:bound_reconstruction}}, the estimation becomes reliable after $T$ units of time. The value of $T$ typically depends on specific attack strategies, so it is difficult to know it unless we have some extra information about the attacker. This point can be illustrated by a toy example: For a positive real number $r$, consider $h_n = r + n$ and $\tau_n = \frac{1}{2}$. Then, we have $T \geqslant r$, because no attack data is available before $t=r$ (no attack occurs until $t = h_1 = r+1$). In this example, $T$ can be arbitrary large as $r$ increases, and thus, we cannot determine its value unless we have prior information on $r$.

In this subsection, we consider a DoS sequence $\xi$ satisfying the following two assumptions which can lead to a quantitative description of $T$.
\begin{assumption} \label{assumption:lower_duration(R1)}
    There exist two constants $b_d \in [0,\inf \mathcal{D}(\xi)]$ and $\kappa'(b_d) > 0$, such that
    $$\lvert \Xi (0,t) \rvert \geqslant -\kappa' + b_d \, t, \quad \forall \, t \geqslant 0,$$
    where $\mathcal{D}(\xi)$ is defined above \textit{Proposition \ref{prop:sharp_bounds}}.
\end{assumption}
\begin{assumption} \label{assumption:lower_frequency(R1)}
    There exist two constants $b_f \in [0,\inf \mathcal{F}(\xi)]$ and $\Lambda'(b_f) > 0$, such that
    $$n_{\xi} (0,t) \geqslant -\Lambda' + b_f \, t, \quad \forall \, t \geqslant 0,$$
    where $\mathcal{F}(\xi)$ is defined above \textit{Proposition \ref{prop:sharp_bounds}}.
\end{assumption}

The following result provides an upper bound of $T$.
\begin{theorem} \label{theorem:upper_bound_of_T}
    Let the estimation $\hat{B}_d (t)$ and $\hat{B}_f (t)$ be generated by the \textit{DoS estimator} (\ref{align:hatT_A(n)})--(\ref{align:hattau_A(n)}). Then, under \textit{Assumptions \ref{assumption:not_edge_case} }to\textit{ \ref{assumption:lower_frequency(R1)}}, the following statements about $T$ are true,
    where $T$ is defined in \textit{Theorem \ref{theorem:bound_reconstruction}}.
    \begin{enumerate}
        \item When $\inf \mathcal{F}(\xi) = 0$, there holds $T=0$.
        \item When $\inf \mathcal{F}(\xi) > 0$ and
    \begin{align} \label{align:quantitative_theo_theta}
        \theta < \min\left\{ \frac{1-\inf \mathcal{D}(\xi)}{1-b_d}, \, \frac{b_f}{\inf \mathcal{F}(\xi)} \right\},
    \end{align}
    there holds
        $T \leq h_{N_1} + \tau_{N_1}$, where $N_1$ is the minimum integer satisfying
    \begin{align} \label{align:quantitative_theo_bd}
    h_{N_1} + \tau_{N_1} > \frac{ \kappa' \theta}{\theta b_d + 1 - \theta - \inf \mathcal{D}(\xi)}
    \end{align}
    and
    \begin{align} \label{align:quantitative_theo_bf}
    h_{N_1} > \frac{\Lambda'}{b_f - \theta \inf \mathcal{F}(\xi)}.
    \end{align}
    \end{enumerate}
\end{theorem}
\begin{proof}
When $\inf \mathcal{F}(\xi) = 0$, there are only finite DoS time intervals $[h_n, h_n + \tau_n)$. In this case, we have $\inf \mathcal{D}(\xi) = 0$ due to \textit{Assumption \ref{assumption:not_edge_case}}. Consequently, the estimation $\hat{B}_d (t)$ \& $\hat{B}_f (t)$ generated by (\ref{align:hatT_A(n)})--(\ref{align:hattau_A(n)}) is always reliable, as $\hat{B}_d (t) > 0$ and $\hat{B}_f (t) > 0$ consistently hold. This validates statement 1).

For the statement 2), the integer $N_1$ is well-defined by (\ref{align:quantitative_theo_bd})--(\ref{align:quantitative_theo_bf}). This is because there are infinitely many DoS time intervals when $\inf \mathcal{F}(\xi) > 0$, and under \textit{Assumption \ref{assumption:not_edge_case}}, there holds $\lim_{n \rightarrow \infty} (h_n + \tau_n) = +\infty$.

We now prove that $\hat{B}_d (t) \in \mathcal{D}(\xi)$ for any $t \geqslant h_{N_1} + \tau_{N_1}$ when $\inf \mathcal{F}(\xi) > 0$. Due to the monotonicity of $\hat{B}_d (t)$, this is equivalent to
    \begin{align} \label{align:quantitative_proof_1}
    \hat{B}_d (h_{N_1} + \tau_{N_1}) > \inf \mathcal{D}(\xi).
    \end{align}
    According to (\ref{align:hatT_A(n)}) and \textit{Assumption \ref{assumption:lower_duration(R1)}}, we have
    \begin{align*}
        \hat{B}_d (h_{N_1} + \tau_{N_1})
        &\geqslant \theta \frac{\lvert \Xi(0,h_{N_1} + \tau_{N_1}) \rvert}{h_{N_1} + \tau_{N_1}} + (1-\theta) \\
        &\geqslant \theta b_d - \frac{\kappa' \theta }{h_{N_1} + \tau_{N_1}} + (1-\theta).
    \end{align*}
    This inequality, combined with (\ref{align:quantitative_theo_theta}) and (\ref{align:quantitative_theo_bd}), implies (\ref{align:quantitative_proof_1}).

    On the other hand, by using (\ref{align:hattau_A(n)}), (\ref{align:quantitative_theo_theta}), (\ref{align:quantitative_theo_bf}), and \textit{Assumption \ref{assumption:lower_frequency(R1)}}, we analogously conclude that $\hat{B}_f (t) \in \mathcal{F}(\xi)$ for any $t \geqslant h_{N_1} + \tau_{N_1}$. Consequently, when $\inf \mathcal{F}(\xi) > 0$, the estimation generated by (\ref{align:hatT_A(n)})--(\ref{align:hattau_A(n)}) becomes reliable no later than $t = h_{N_1} + \tau_{N_1}$. This validates statement 2), and the proof is completed. \hfill $\square$
\end{proof}

{

}

\section{Anti-DoS consensus of MASs} \label{sec:consensus}

In this section, we show how \textit{DoS estimator} (\ref{align:hatT_A(n)})--(\ref{align:hattau_A(n)}) can help regulate a multi-agent system (MAS) to reach an \textit{average consensus} \cite{Saber_average_consensus_ACC2003,JBerneburg_consensus-etc_TAC2021} without pre-assuming the duration and frequency parameters of the attacker.

\subsection{Problem formulation} \label{subsec:consensus_problem_formulation}
We consider the following MAS
\begin{align} \label{align:MAS}
    \dot{x}_i (t) = u_i (t), \; t\geqslant 0,\, i = 1,2,\cdots,N,
\end{align}
where $x_i \in \mathbb{R}$ represents the state of the $i$-th agent and $u_i \in \mathbb{R}$ represents the control input to be designed. A common setting of $u_i$ is as follows \cite{GXie_sample-consensus_ACC2009,HSu_consensus_tac2023}:
\begin{align*}
    u_i (t) = \sum_{j=1}^N a_{ij} \left( x_j(t_k) - x_i(t_k) \right),\, t \in [t_k,t_{k+1}),
\end{align*}
where $a_{ij} = 1$ if there is a connection between agents $i$ and $j$, or $a_{ij} = 0$ otherwise. In addition, $\{t_k\}_{k \in \mathbb{N}_+}$ is the sequence of sampling instants.
Due to the impact of the DoS attack, some transmissions are denied, and the actual input of the agent $i$ therefore becomes
\begin{align} \label{align:MAS_ui}
    u_i (t) = 
    \left\{
    \begin{array}{ll}
        0, &  t_k \in \Xi(0,+\infty),\\
        \sum_{j=1}^N a_{ij} \left(x_j(t_k) - x_i(t_k)\right), & t_k \in \Theta(0,+\infty)
    \end{array}
    \right.
\end{align}
for any $t \in [t_k,t_{k+1})$. Here, the notations $\Xi(\tau,s)$ and $\Theta(\tau,s)$ are, respectively, defined in (\ref{align:Xi}) and (\ref{align:Theta}).

For the convenience of the statement, we denote by $\mathscr{L} = (l_{ij})_{N \times N}$ the Laplacian of the multi-agent network, where
\begin{eqnarray*}
    l_{ij} = 
    \left\{
    \begin{array}{lcl}
        -a_{ij}, & i \neq j, \\
        \sum_{i \neq j} a_{ij}, & i = j. 
    \end{array}
    \right.
\end{eqnarray*}
We assume $\mathscr{L}$ to satisfy the following assumption.
\begin{assumption} \label{assumption:undirected_connected}
    \cite{Saber_average_consensus_ACC2003} $\mathscr{L}$ is symmetric and irreducible. As a result, the eigenvalues of $\mathscr{L}$ can be arranged as $$0 = \lambda_1 (\mathscr{L}) < \lambda_2 (\mathscr{L}) \leqslant \lambda_3 (\mathscr{L}) \leqslant \cdots \leqslant \lambda_N(\mathscr{L}).$$
\end{assumption}

Our control objective is to make MAS (\ref{align:MAS}) reach an average consensus, which is defined as follows.
\begin{definition} \label{definition:average_consensus}
    \textit{(\cite{Saber_average_consensus_ACC2003} Average consensus)} MAS (\ref{align:MAS}) is said to reach average consensus, if
    $$\lim_{t \rightarrow +\infty}{\left(x_i (t) - \frac{1}{N} \sum_{j=1}^{N} x_{j}(0)\right)} = 0,\; \forall \, i = 1,\cdots,N.$$
    \hfill $\square$
\end{definition}

In a classical sampled data setting, the sampling is usually carried out periodically, that is, $t_{k+1} - t_k = \Delta$ with some constant $\Delta>0$. Under such a setting, average consensus can be reached if $\lvert 1 - \Delta 
\lambda_N (\mathscr{L}) \rvert < 1$ \cite{YYou_discreteMAS_TAC2011} and 
    \begin{align} \label{align:Delta_known_capacity}
    B_d + \Delta B_f < 1\text{ \cite{WXu_distributedDoS_tcyb2020}}. 
    \end{align}
However, the design of $\Delta$ in (\ref{align:Delta_known_capacity}) is seldom reliable because it relies on presupposed values of the attacker's \textit{duration-bound} $B_d$ \& \textit{frequency-bound} $B_f$.

Our task in this section is to design the sampling sequence $\{t_k\}_{k \in \mathbb{N}}$, such that control law (\ref{align:MAS_ui}) drives MAS (\ref{align:MAS}) to reach an average consensus.
At any sampling time $t_k$, only past knowledge about the attacker (namely $\mathscr{K}_a (0,t_k)$) is available.
Our basic idea is letting $t_{k+1} - t_k$ be \textit{adaptive} to $\hat{B}_d (s)$ \& $\hat{B}_f (s)$ (which are functions of $\mathscr{K}_a (0,t_k)$), where $s\leqslant t_k$, so that it converges to a $\Delta$ satisfying (\ref{align:Delta_known_capacity}). Such adaptive sampling intervals should have a strictly positive lower bound to guarantee a finite sampling rate. If we denote $\Delta_k \coloneqq t_{k+1} - t_k$, this task can be written more formally as follows.

\begin{problem} \label{problem:average_consensus}
\textit{(Consensus design of the sampling sequence)}    Design the sampling interval $\Delta_k$ as
    $$\Delta_k = \Delta_k(\mathscr{K}_a (0,t_k)),$$
    such that
    \begin{itemize}
        \item[a)] Control law (\ref{align:MAS_ui}) regulates MAS (\ref{align:MAS}) to reach an average consensus.
        \item[b)] There exists a constant $\underline{\Delta} > 0$ satisfying
        $$\Delta_k \geqslant \underline{\Delta},\; \forall \, k \in \mathbb{N}.$$
    \end{itemize}
\end{problem}

\subsection{Solution to \textit{Problem \ref{problem:average_consensus}}: The adaptive sampling sequence}

This subsection gives a detailed design of $\Delta_k$, which solves \textit{Problem \ref{problem:average_consensus}}. We let $t_1 = 0$ and initialize $\Delta_k$ as
\begin{align} \label{align:Delta_k_initialize}
    \Delta_k &= \Delta_0,\,t_k \in [0,h_1+\tau_1),
\end{align}
where $\Delta_0$ satisfies $\lvert 1 - \Delta_0 \lambda_N (\mathscr{L}) \rvert < 1$. When $t = h_n +\tau_n$, $n \in \mathbb{N}_+$, we update $\Delta_k$ as
\begin{align} \label{align:Delta_k}
    \Delta_k &=
    \min\left\{\Delta_0, \frac{1}{\gamma_1\hat{B}_f (h_n)}\left(1- \hat{B}_d (h_n + \tau_n)\right) \right\},\nonumber \\
    &\qquad \qquad \qquad \qquad t_k \in [h_n + \tau_n, h_{n+1} + \tau_{n+1}).
\end{align}
Here, $\hat{B}_d (t)$ and $\hat{B}_f (t)$ are generated by the \textit{DoS estimator} (\ref{align:hatT_A(n)})--(\ref{align:hattau_A(n)}), and $\gamma_1 > 1$ ensures that $\Delta_k$ satisfies an ``estimated'' version of (\ref{align:Delta_known_capacity}), that is,
$$\hat{B}_d (h_n) + \Delta_k \hat{B}_f (h_n) < 1.$$

The following result shows that the sampling logic determined by (\ref{align:Delta_k_initialize}) and (\ref{align:Delta_k}) can well address \textit{Problem \ref{problem:average_consensus}}.

\begin{theorem} \label{theorem:consensus_MAS_adaptive_sampling}
    Under \textit{Assumptions \ref{assumption:not_edge_case}} and \textit{\ref{assumption:undirected_connected}}, MAS (\ref{align:MAS})--(\ref{align:MAS_ui}) subject to a DoS sequence $\xi$ can reach average consensus if the sampling intervals $\{\Delta_k\}_{k \in \mathbb{N}_+}$ are determined by (\ref{align:Delta_k_initialize})--(\ref{align:Delta_k}). In addition, $\Delta_k$ has a lower bound $\underline{\Delta}>0$.
\end{theorem}
\begin{proof}
    See Appendix \ref{appC}. \hfill $\square$
\end{proof}
\begin{remark}
For consensus of MASs, \textit{Theorem \ref{theorem:consensus_MAS_adaptive_sampling}} significantly extends the reach of the security design (\ref{align:Delta_known_capacity}) \cite{WXu_distributedDoS_tcyb2020}. In particular, even if $B_d$ or $B_f$ is previously unknown to the defender (and therefore (\ref{align:Delta_known_capacity}) cannot be directly implemented), satisfactory resilience is still preserved if we design the sampling sequence as in (\ref{align:Delta_k_initialize})--(\ref{align:Delta_k}). The theorem is, therefore, an illustration of how we broaden the applicability of existing anti-DoS control strategies based on our new \textit{DoS estimator}.

\end{remark}

\section{Anti-DoS impulsive stabilization of nonlinear systems} \label{sec:impulsive_stabilization}
In this section, we show that the \textit{DoS estimator} is also effective in the impulsive stabilization of nonlinear systems \cite{TYang_impulsive-control_TAC1999,XLi_event_impulsive_control_tac2020,ZGLi_impulsive_control_system_tac2001} in the presence of previously unknown DoS attacks.

To avoid overcomplicating notations, the symbols $t_k$ and $\Delta_k$
in this section will represent different physical meanings compared to the previous section. Specifically, this section uses $t_k$ to represent the $k$-th impulsive \textit{control instant} and uses $\Delta_k \coloneqq t_{k+1} - t_k$ to represent the $k$-th \textit{control interval}.

\subsection{Problem formulation} \label{subsec:impulsive_stabilization_problem_formulation}
We consider the following nonlinear system \cite{TYang_impulsive-control_TAC1999}
\begin{eqnarray} \label{align:nonlinear_system}
\left\{
\begin{array}{lcl}
    \dot{X}(t) = f(t,X(t)), & t \neq t_k, \\
    X(0^-) = X_0. &
\end{array}
\right.
\end{eqnarray}
with an impulsive control law
\begin{align} \label{align:impulsive_control}
    X(t_k^+) = X(t_k^-) + U(X(t_k^-)),\, t_k \in \Theta(0,+\infty).
\end{align}
Here, $X \in \mathbb{R}^n$ is a right-continuous state variable, $f: \mathbb{R}_{\geqslant 0} \times \mathbb{R}^n \mapsto \mathbb{R}^n$ is continuous, $f(t,0) = 0,\,\forall \, t \geqslant 0$, and $U(\cdot)$ is the impulsive control input.
Due to the DoS attack, the impulsive input cannot be operated over $\Xi(0,+\infty)$, and therefore
\begin{align} \label{align:impulsive_invalid_control}
    X(t_k^+) = X(t_k^-),\,t_k \in \Xi(0,+\infty).
\end{align}
Impulsive control system (\ref{align:nonlinear_system})--(\ref{align:impulsive_invalid_control}) is assumed to satisfy the following assumption.
\begin{assumption} \label{assumption:impulsive_system}
    \cite{book_impulsive_differential_equation,TYang_impulsive-control_TAC1999} There exists a function $V:\mathbb{R}_{\geqslant 0} \times \mathbb{R}^n \mapsto \mathbb{R}_{\geqslant 0}$, such that
    \begin{enumerate}
        \item $V$ is continuous in $[t_{k},t_{k+1}) \times \mathbb{R}^n$, and for each $X \in \mathbb{R}^n$ and $k = 1,2,\cdots$, the limit
        \begin{align*}
            \lim_{(t,y) \rightarrow (t_k^-,X)} V(t,Y) = V(t_k,X)
        \end{align*}
        exists.
        \item $V$ is locally Lipschitzian in its second argument.
        \item There exists a $\beta>0$, such that for any $X \in \mathbb{R}^n$ and $ t \neq t_k$, there holds that
        \begin{align*}
        D^+ V(t,X) \leqslant \beta V(t,X).
        \end{align*}
        \item There exists a $\mu \in (0,1)$ such that
        \begin{align*}
            V(t,X+ U(t,X)) \leqslant \mu V(t^-,X),\,\forall \, X \in \mathbb{R}^n,\, t = t_k.
        \end{align*}
        \item There exist functions $a(\cdot),b(\cdot) \in \mathcal{K}$, such that $a(\lvert X \rvert) \leqslant V(t,X) \leqslant b(\lvert X \rvert)$ holds on $\mathbb{R}_{\geqslant 0} \times \mathbb{R}^n$.
    \end{enumerate}
\end{assumption}

Our control objective is to stabilize the system (\ref{align:nonlinear_system}) exponentially, which can be defined as follows.
\begin{definition} \label{definition:exponential_stabilization}
    \textit{(\cite{Khalil_nonlinear} Exponential stabilization)} It is said that impulsive controller (\ref{align:impulsive_control})--(\ref{align:impulsive_invalid_control}) exponentially stabilizes system (\ref{align:nonlinear_system}), if there exist constants $C_0 > 0$ and $\zeta>0$ satisfying
    $$\lvert X(t) \rvert \leqslant C_0 \lvert X_0 \rvert e^{-\zeta t},\; \forall \, t \geqslant 0.$$
\end{definition}

Our task in this subsection is to design $\Delta_k$ to achieve the objective as defined in \textit{Definition \ref{definition:exponential_stabilization}} above. Similar to Section \ref{sec:consensus}, we can only use the past attacker information $\mathscr{K}_a(0,t_k)$. Furthermore, $\Delta_k$ should have a strictly positive lower bound. This task can be written more formally as follows.

\begin{problem} \label{problem:impulsive_stabilization}
\textit{(Stabilization design of the impulsive control sequence)}    Design the control interval $\Delta_k$ as
        $$\Delta_k = \Delta_k(\mathscr{K}_a (0,t_k)),$$
    such that
    \begin{itemize}
        \item[a)] Impulsive controller (\ref{align:impulsive_control})--(\ref{align:impulsive_invalid_control}) exponentially stabilizes system (\ref{align:nonlinear_system}).
        \item[b)] There exists a constant $\underline{\Delta} > 0$ satisfying
        $$\Delta_k \geqslant \underline{\Delta},\; \forall \, k \in \mathbb{N}.$$
    \end{itemize}
\end{problem}

\subsection{Solution to \textit{Problem \ref{problem:impulsive_stabilization}}: The adaptive control sequence}
Now, we provide a solution to \textit{Problem \ref{problem:impulsive_stabilization}}.
We let $t_1 = 0$ and initialize $\Delta_k$ as
\begin{align} \label{align:impulsive_Delta_k_initialize}
    \Delta_k &= \Delta_0 \coloneqq -\frac{\chi}{\gamma_3 \beta},\,t_k \in [0,h_1+\tau_1),
\end{align}
where $\chi \coloneqq \ln \mu < 0$, $\gamma_3 > 1$ is a constant, $\mu<1$ and $\beta > 0$ are given in \textit{Assumption \ref{assumption:impulsive_system}}. When $t = h_n + \tau_n$, $n \in \mathbb{N}_+$, $\Delta_k$ is updated as
\begin{align} \label{align:impulsive_Delta_k}
    \Delta_k &= 
    \frac{\chi \left(1-\hat{B}_d (h_n + \tau_n)\right)}{\gamma_3 \left( \hat{B}_f (h_n)\chi - \beta \right)},\nonumber\\
    &\qquad \qquad \qquad t_k \in [h_n + \tau_n, h_{n+1} + \tau_{n+1}).
\end{align}
where $\hat{B}_d (t)$ and $\hat{B}_f (t)$ are generated by (\ref{align:hatT_A(n)})--(\ref{align:hattau_A(n)}).

The following result shows that (\ref{align:impulsive_Delta_k_initialize})--(\ref{align:impulsive_Delta_k}) can well address \textit{Problem \ref{problem:impulsive_stabilization}}.

\begin{theorem} \label{theorem:impulsive_stabilization_adaptive_interval}
    Under \textit{Assumptions \ref{assumption:not_edge_case}} and \textit{\ref{assumption:impulsive_system}}, nonlinear system (\ref{align:nonlinear_system})--(\ref{align:impulsive_invalid_control}) subject to a DoS sequence $\xi$ can be exponentially stabilized if the control intervals $\{\Delta_k\}_{k \in \mathbb{N}_+}$ are determined by (\ref{align:impulsive_Delta_k_initialize})--(\ref{align:impulsive_Delta_k}). In addition, $\Delta_k$ has a lower bound $\underline{\Delta}>0$.
\end{theorem}
\begin{proof}
    See Appendix \ref{appD}. \hfill $\square$
\end{proof}

\section{Numerical simulations} \label{sec:numerical_examples}
    In this section, we consider a DoS sequence $\xi$ with $\inf \mathcal{D} (\xi) = 2/3$ and $\inf \mathcal{F} (\xi) = 0.5$. We simulate a scenario where the attacker performs weakly on the time interval $[0,12]$, and show that the \textit{DoS estimator} (\ref{align:hatT_A(n)})--(\ref{align:hattau_A(n)}) will not be misled by this initial weak performance. In fact, a reliable estimation can still be obtained after $t = 12$, when the attack intensifies and gradually becomes consistent with the true \textit{duration} \& \textit{frequency bounds}.

    The two control scenarios in Sections \ref{sec:consensus} and \ref{sec:impulsive_stabilization} are studied in this part. Our new estimator-based controllers have satisfactory control performances in both cases.
    Particularly in the impulsive stabilization scenario, we provide two simulation results to justify the discussion in \textit{Remarks \ref{remark:parameter_tuning} }and\textit{ \ref{remark:epsilon0}}. The trade-off between security and cost-efficiency is clearly manifested by the extra simulation.
    
    The end of this section considers a special DoS sequence $\xi$ as a counterexample, illustrating the necessity of the design $\theta<1$, as discussed in \textit{Remark \ref{remark:theta}}.

\begin{example} \label{ex1}
    In this example, both scenarios in Sections \ref{sec:consensus} and \ref{sec:impulsive_stabilization} share the same \textit{DoS estimator}. In particular, we let the estimated \textit{duration} \& \textit{frequency bounds} be generated by (\ref{align:hatT_A(n)})--(\ref{align:hattau_A(n)}) with $\epsilon_0 = 0.01$, $\theta = 0.67$ and $\ell = 2$. Fig. \ref{fig:capacity_estimation_0.5} shows the estimation's performance. As can be seen, $\hat{B}_d (t)$ converges to $\mathcal{D} (\xi) = [2/3,1]$ after $n = 10$, and $\hat{B}_f (t)$ converges to $\mathcal{F} (\xi) = [0.5,+\infty]$ after only $2$ DoS launches.

    \begin{figure}[ht]
        \centering
        \includegraphics[width = 8.3cm]{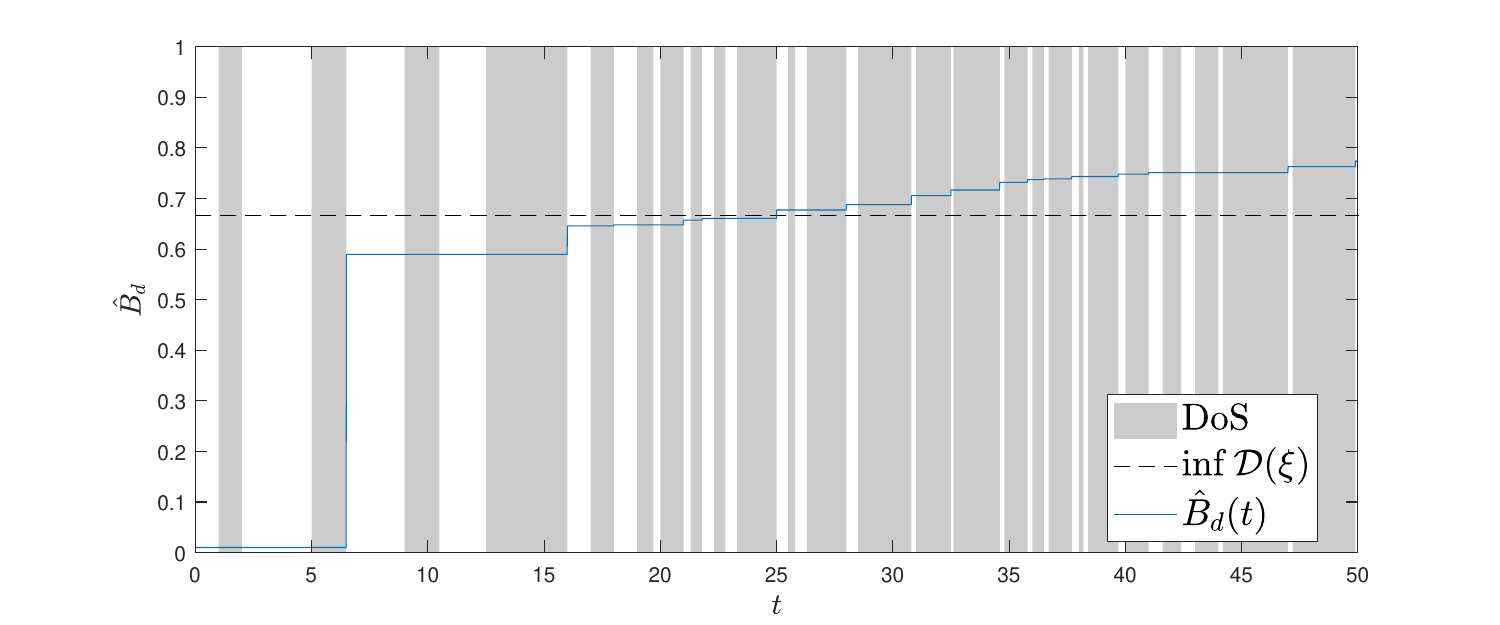}
        \includegraphics[width = 8.3cm]{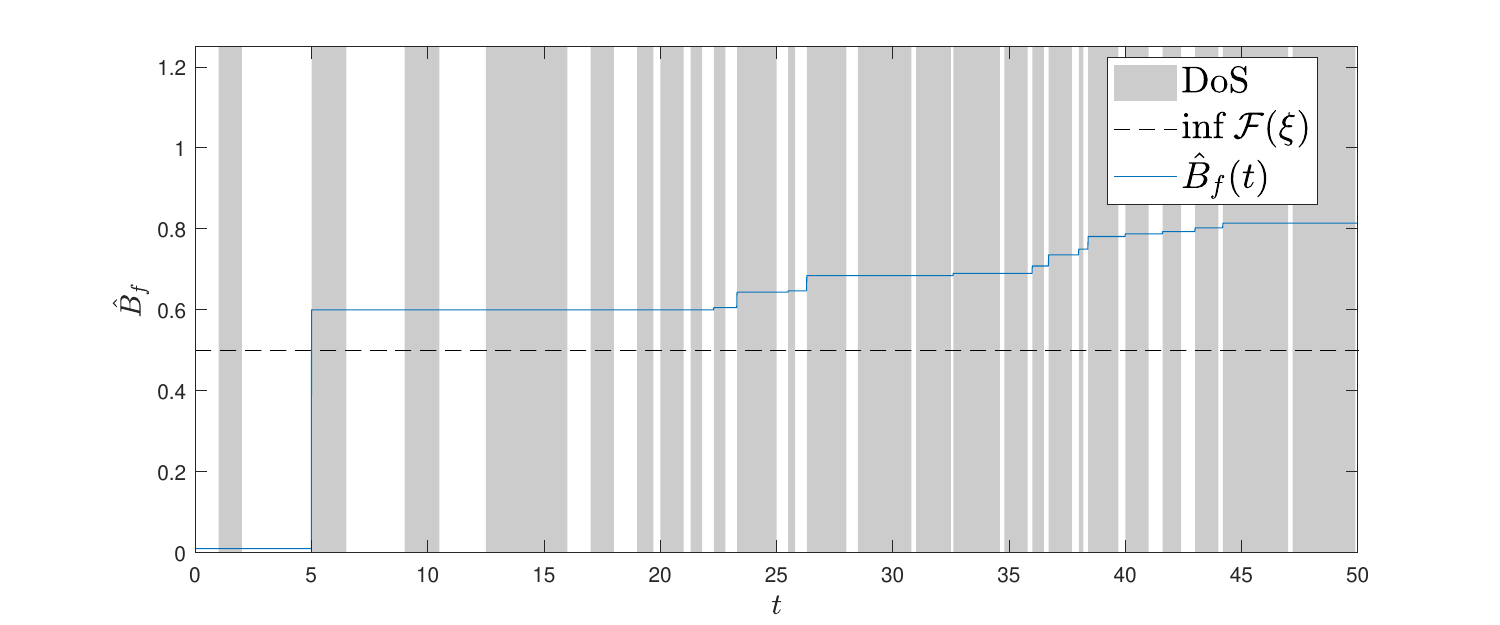}
        \caption{Top: The estimated \textit{duration-bound} $\hat{B}_d (t)$ is generated by (\ref{align:hatT_A(n)}). Bottom: The estimated \textit{frequency-bound} $\hat{B}_f (t)$ is generated by (\ref{align:hattau_A(n)}). (With $\theta = 0.67$ and $\ell = 2$.)}
        \label{fig:capacity_estimation_0.5}
    \end{figure}

    Now, we configure security controllers based on the results in Sections \ref{sec:consensus} and \ref{sec:impulsive_stabilization}. First, we study the sampled data control performance of MAS (\ref{align:MAS}) with $7$ agents and a \textit{ring topology}. By simple calculation, $\lambda_N (\mathscr{L}) = 3.8019$, and thus $\Delta_0$ should be strictly less than $0.5260$ according to \textit{Theorem \ref{theorem:consensus_MAS_adaptive_sampling}}. In addition, the initial values are randomly chosen {from} $[-10,10]$, and the average consensus value is $-3/7$. By taking $\Delta_0 = 0.4208$, $\gamma_1 = 1.3$ and letting $\Delta_k$ be generated by (\ref{align:Delta_k_initialize})--(\ref{align:Delta_k}), average consensus can be achieved as shown at the top of Fig. \ref{fig:consensus_MAS}. At the bottom of Fig. \ref{fig:consensus_MAS}, we plot the adaptive sampling intervals $\Delta_k$ generated by (\ref{align:Delta_k_initialize})--(\ref{align:Delta_k}), as well as its theoretical supremum determined by \textit{Theorem \ref{theorem:consensus_MAS_adaptive_sampling}}. As can be seen, the generated sampling sequence is over two times more frequent than necessary. Hence, this suggests that our strategy is somewhat conservative, which is consistent with our intuition because we always {consider a stronger opponent than observed}. We will discuss how to reduce such conservatism later in \textit{Example \ref{ex2}}.

    An interesting observation from Fig. \ref{fig:consensus_MAS} is that the average consensus is achieved before the estimation in Fig. \ref{fig:capacity_estimation_0.5} becomes reliable. This phenomenon can be attributed to the unique dynamics of MASs described by (\ref{align:MAS})--(\ref{align:MAS_ui}): According to the first line of (\ref{align:MAS_ui}), all agents maintain constant states and do not diverge from each other when transmissions are denied by the attack. Conversely, according to the second line, agents converge towards the average value when transmissions are successful. This dual behavior enables the system to achieve average consensus with a sufficient number of successful transmissions, even when the estimator is not yet fully effective.

    \begin{figure}[ht]
        \centering
        \includegraphics[width = 8.3cm]{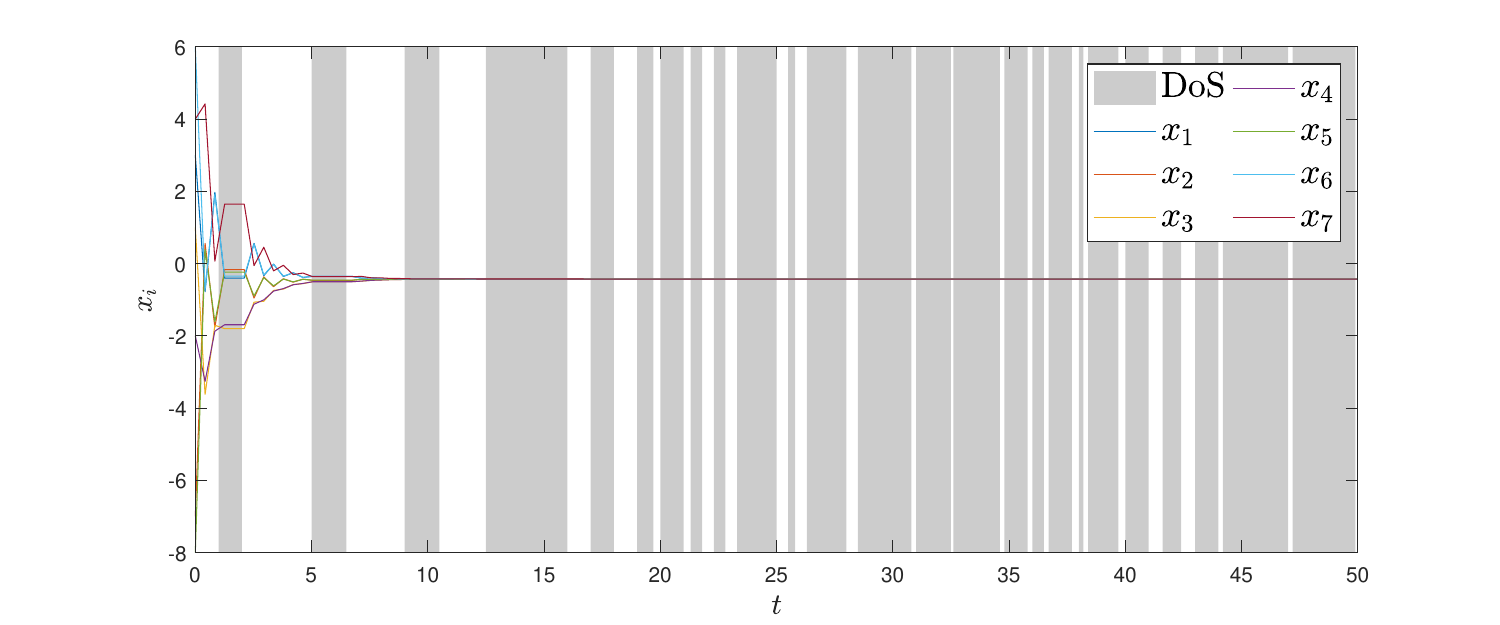}
        \includegraphics[width = 8.3cm]{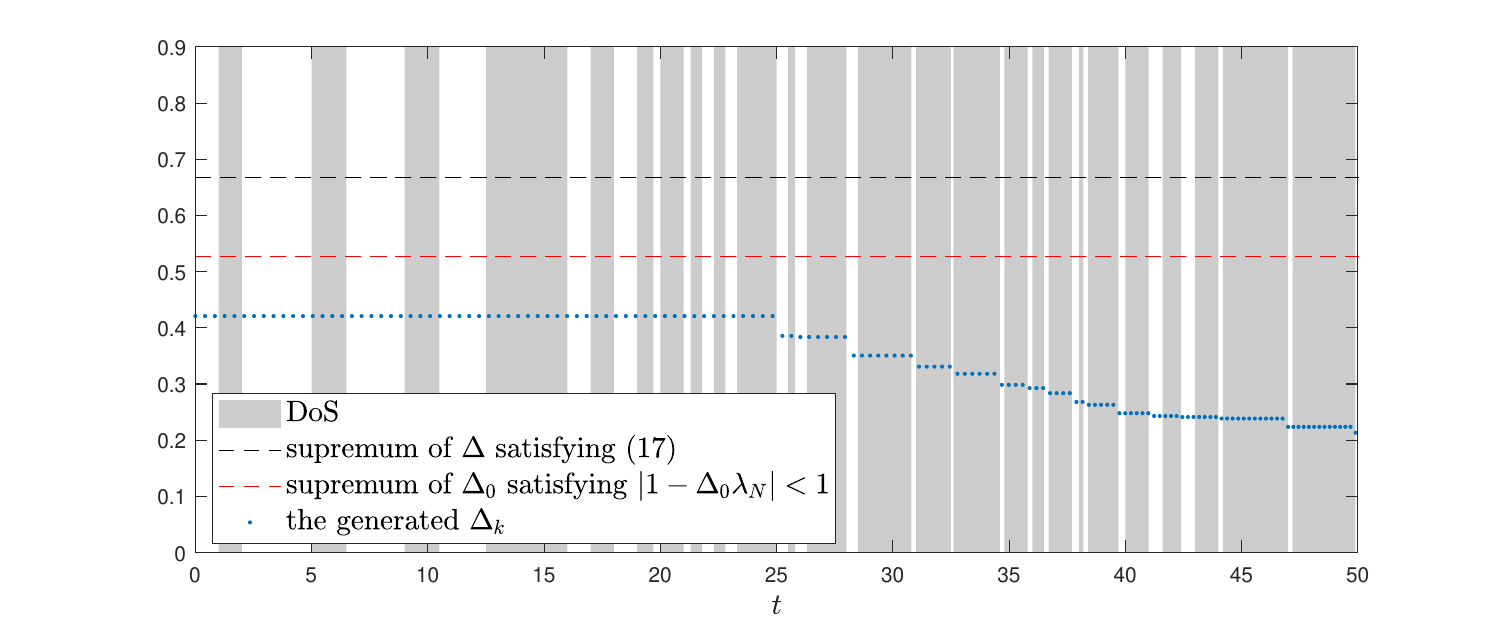}
        \caption{Top: Consensus of MAS (\ref{align:MAS}) regulated by the sample-data controller (\ref{align:MAS_ui}).
        Bottom: The values of $\Delta_k$ generated by (\ref{align:Delta_k_initialize})--(\ref{align:Delta_k}) at each sampling instance, as well as the theoretical supremum of the intervals determined by \textit{Theorem \ref{theorem:consensus_MAS_adaptive_sampling}}.}
        \label{fig:consensus_MAS}
    \end{figure}

    Next, we study the impulsive control performance of the nonlinear system (\ref{align:nonlinear_system}) with a $2$-dimensional state and $f(t,X) = A X$. Here, we have an unstable matrix $A = \begin{bmatrix}
        1 & 0.3 \\
        0 & 1
    \end{bmatrix}$ and the impulsive input $U(X) = 0.7 X$. By choosing $V(t,X) = \lvert X \rvert$, the parameters in \textit{Assumption \ref{assumption:impulsive_system}} are obtained as $\mu = 0.7$ and $\beta = \sqrt{\lambda_{\max} (A^T A)} = 1.1612$. Taking $\gamma_3 = 1.2$ and letting $\Delta_k$ be generated by (\ref{align:impulsive_Delta_k_initialize})--(\ref{align:impulsive_Delta_k}), impulsive stabilization can be achieved as shown at the top of Fig. \ref{fig:stabilization_impulsive_nonlinear_05}. 

    \begin{figure}[ht]
        \centering
        \includegraphics[width = 8.3cm]{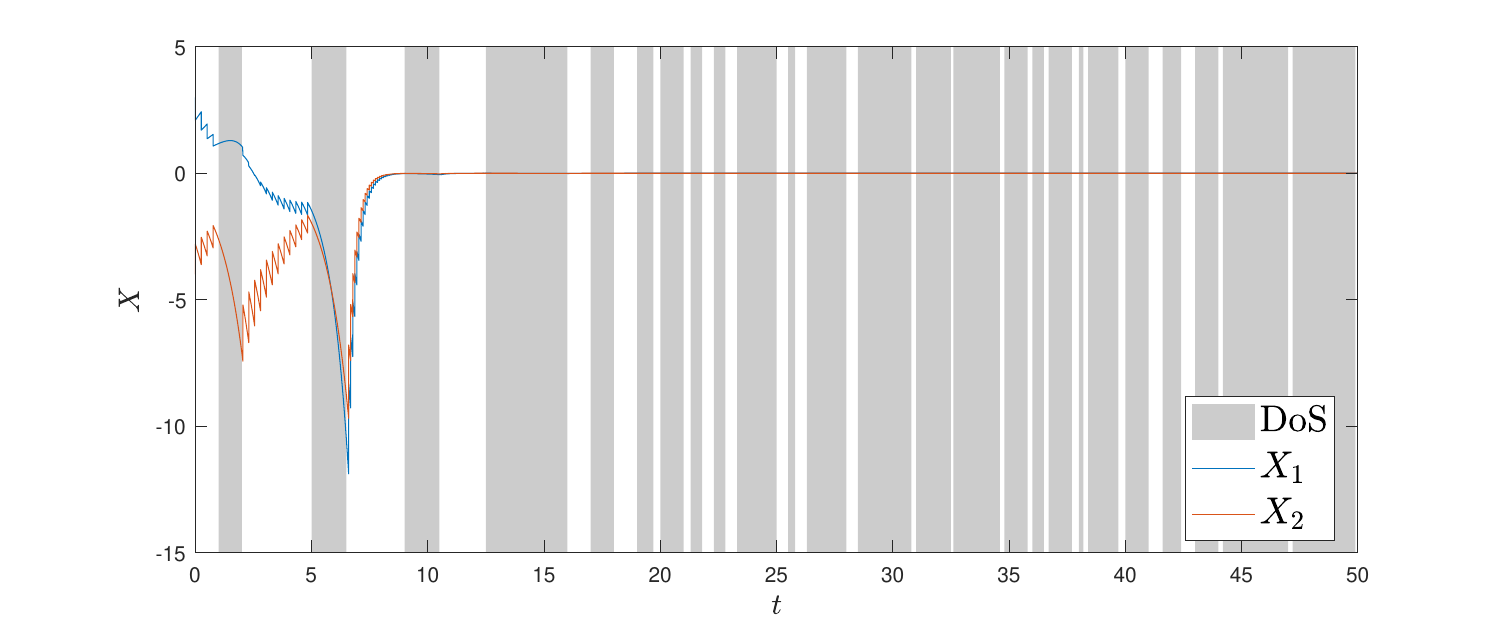}
        \includegraphics[width = 8.3cm]{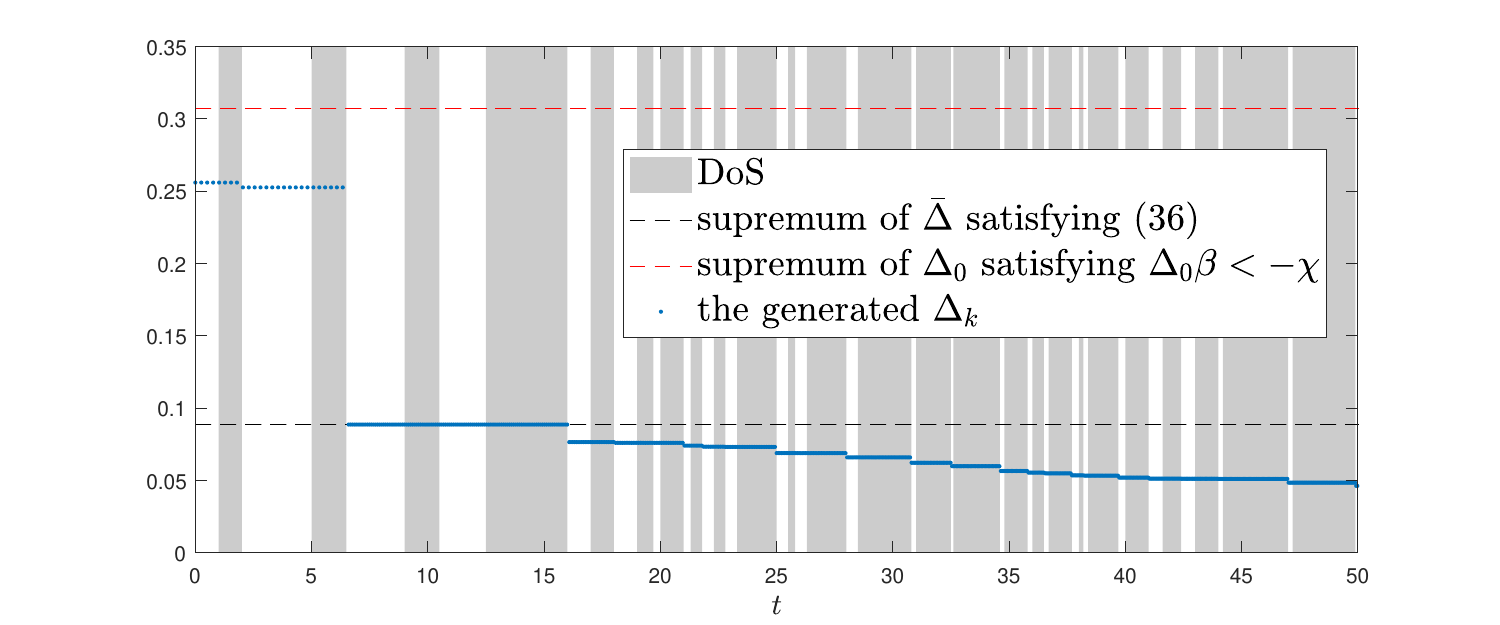}
        \caption{Top: Stabilization of system (\ref{align:nonlinear_system}) regulated by the impulsive controller (\ref{align:impulsive_control})--(\ref{align:impulsive_invalid_control}).
        Bottom: The values of $\Delta_k$ generated by (\ref{align:impulsive_Delta_k_initialize})--(\ref{align:impulsive_Delta_k}) at each control instants, as well as the theoretical supremum of the intervals determined by \textit{Theorem \ref{theorem:impulsive_stabilization_adaptive_interval}}.}
        \label{fig:stabilization_impulsive_nonlinear_05}
    \end{figure}
\end{example}

\begin{example} \label{ex2}
    In this example, the DoS sequence is the same as in \textit{Example \ref{ex1}}, and we consider two other pairs of values of $\theta$ and $\ell$, specifically $(\theta, \ell) = (0.9,2)$ and $(\theta, \ell) = (0.9,3)$.
    As can be seen in Fig. \ref{fig:comparison_estimated_bound}, the estimator generates tighter \textit{duration} \& \textit{frequency bounds} of the DoS sequence $\xi$ when either $\theta$ or $\ell$ increases. This is because a larger $\theta$ means being less skeptical about the historical data $B_d (i)$ and $B_f (i)$, and a larger $\ell$ means waiting for more data before starting the estimation. Both of these adjustments enable the estimated bounds to more closely approximate the actual values, resulting in a more precise characterization of the DoS sequence's behavior.
    
    \begin{figure}[ht]
        \centering
        \includegraphics[width = 8.3cm]{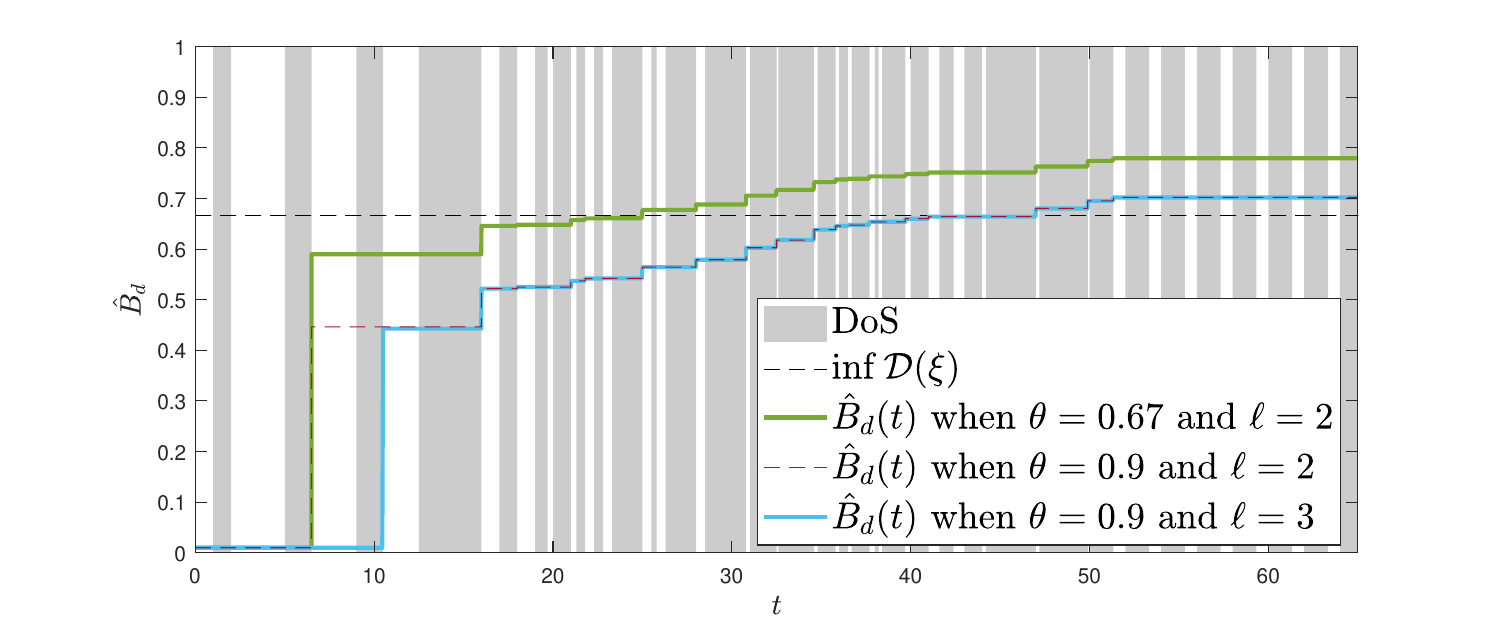}
        \includegraphics[width = 8.3cm]{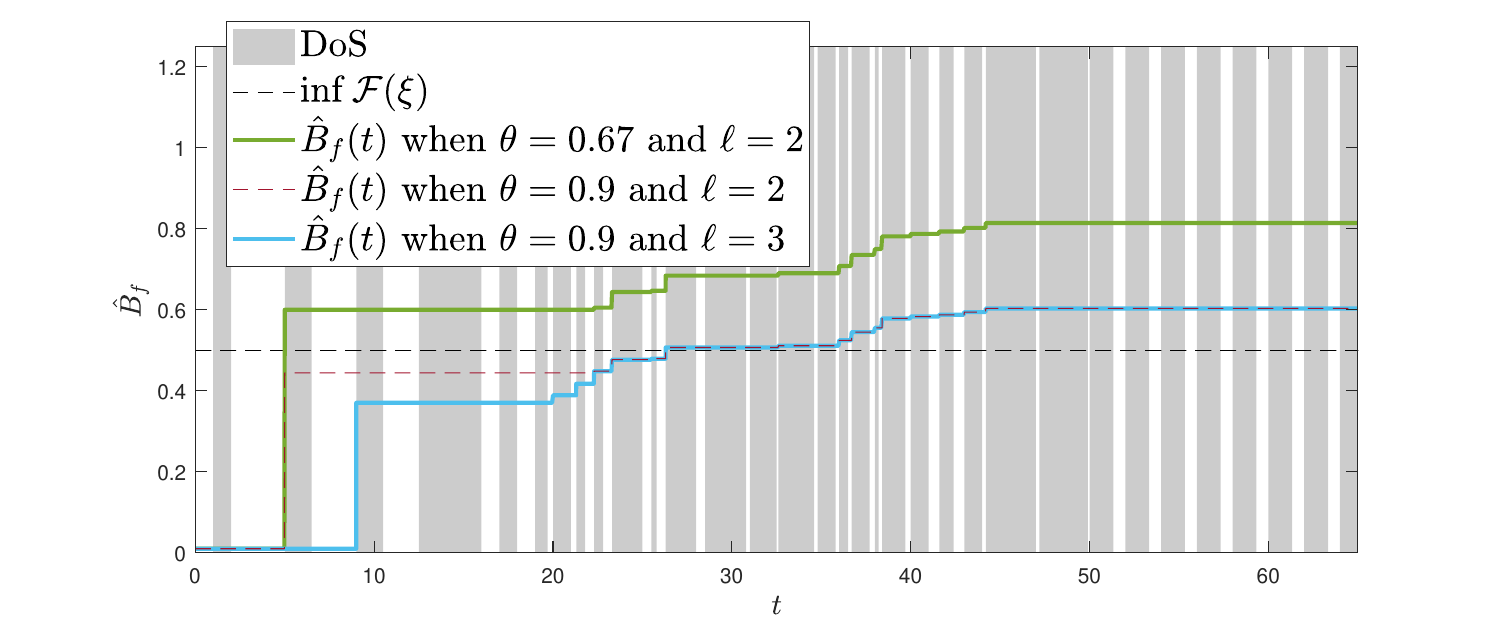}
        \caption{Estimations of $\xi$'s \textit{duration-bound} (top) \& \textit{frequency-bound} (bottom) with different values of $\theta$ and $\ell$.}
        \label{fig:comparison_estimated_bound}
    \end{figure}

    In \textit{Remark \ref{remark:parameter_tuning}}, we have claimed that transient performance might be sacrificed as the trade-off for tighter DoS bounds. In order to demonstrate this, we use the two estimators with $(\theta, \ell) = (0.9,2)$ and $(\theta, \ell) = (0.9,3)$, respectively, to generate the impulsive control sequences and compare the transient responses of the system (\ref{align:nonlinear_system}). For comparison fairness, the system is set to be the same as in \textit{Example \ref{ex1}} except for the control instants $t_k$. As shown in Fig. \ref{fig:stabilization_impulsive_nonlinear_083}, an increase in $\theta$ or $\ell$ can result in larger intervals $\Delta_k$, while the system takes more time to achieve stabilization. Moreover, for the pair $(\theta, \ell) = (0.9,3)$, the state peak is near $500$, which is an abnormal value for many practical systems.
    
    \begin{figure}[ht]
        \centering
        \includegraphics[width = 8.3cm]{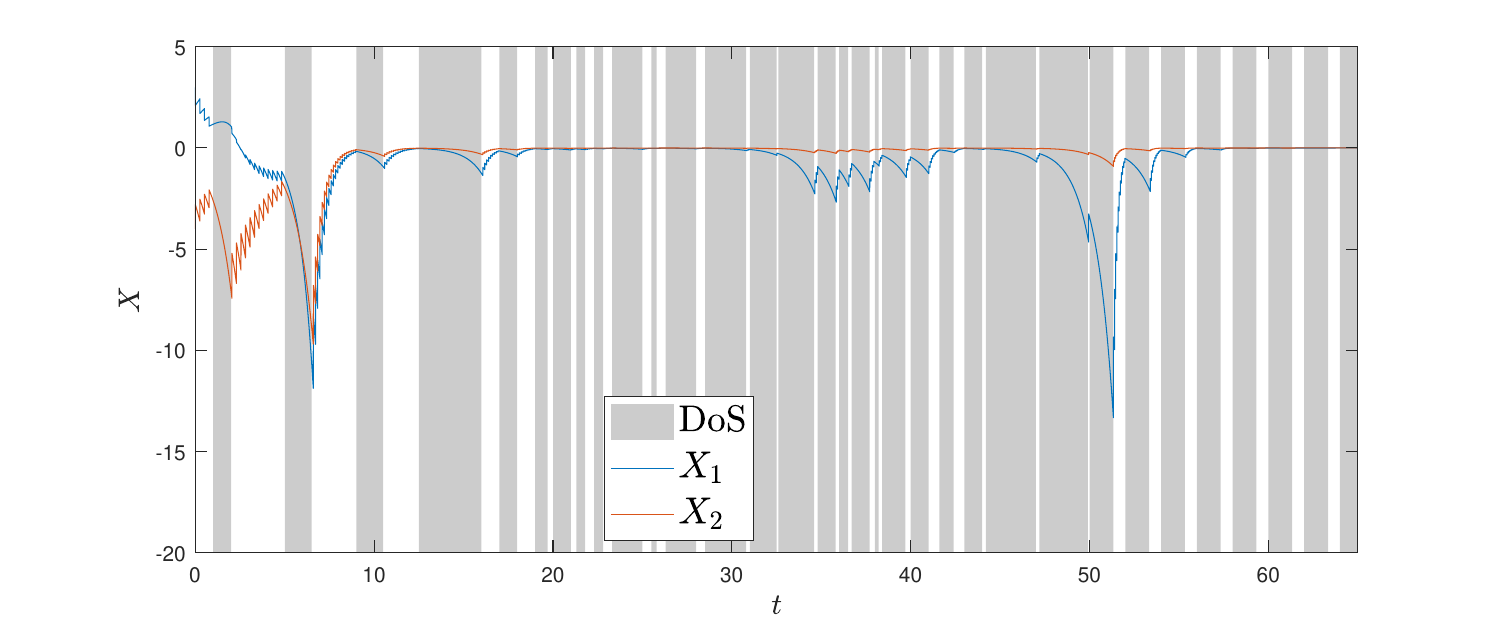}
        \includegraphics[width = 8.3cm]{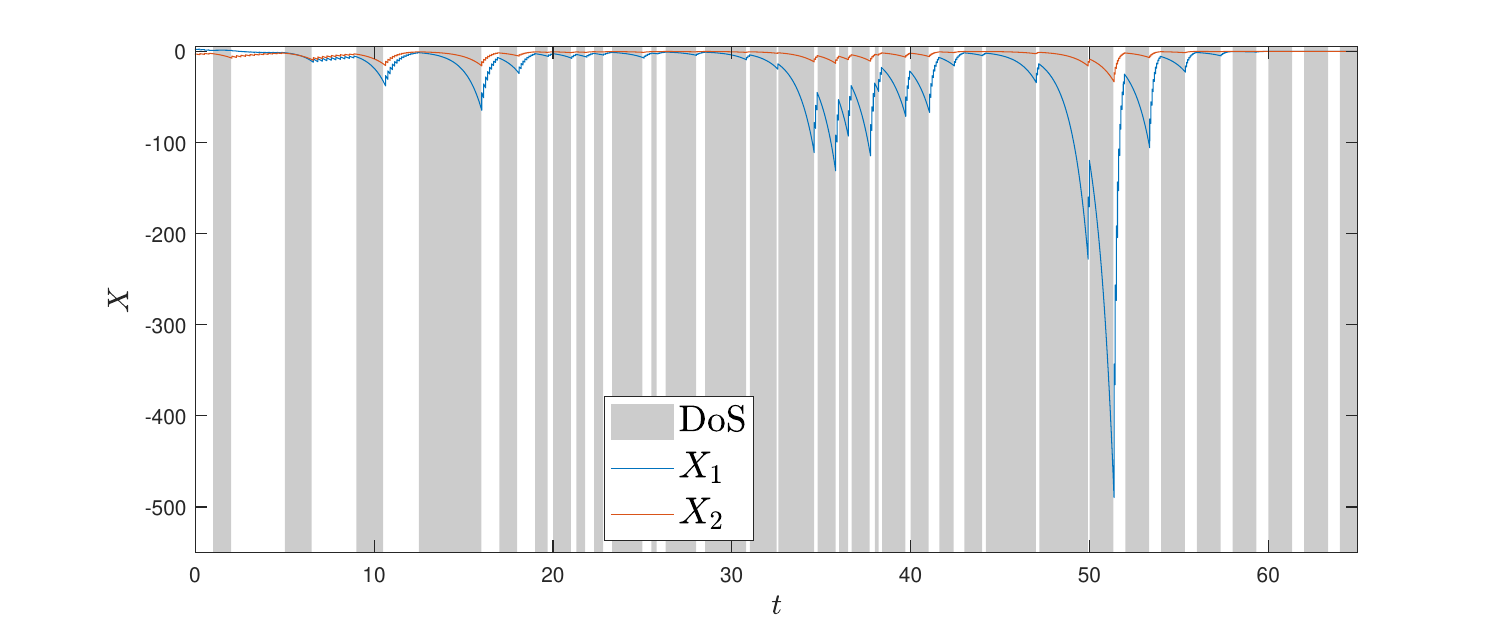}
        \includegraphics[width = 8.3cm]{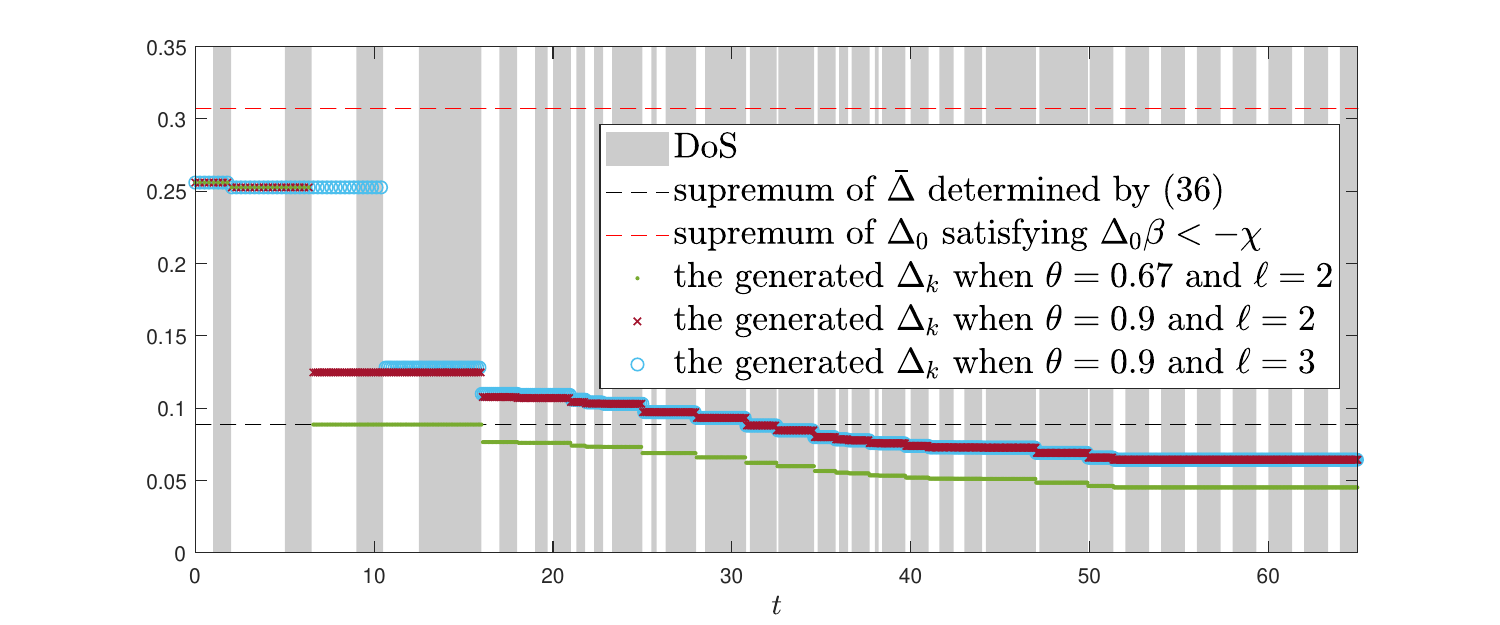}
        \caption{Top and middle: The transient response of system (\ref{align:nonlinear_system}) under impulsive control (\ref{align:impulsive_control})--(\ref{align:impulsive_invalid_control}). The control sequences are generated using \textit{DoS estimators} with $(\theta,\ell) = (0.9,2)$ and $(\theta,\ell) = (0.9,3)$, respectively.
        Bottom: The values of $\Delta_k$ generated at each control instant for different pairs of $(\theta,\ell)$.}
        \label{fig:stabilization_impulsive_nonlinear_083}
    \end{figure}
\end{example}
    
    The reason why we do not compare the case of MAS consensus is that it is not sensitive to the \textit{DoS estimator}. More specifically, we can see from Fig. \ref{fig:consensus_MAS} that
    the range of $\Delta_k$ determined by $\lvert 1 - \Delta_0 \lambda_N(\mathscr{L}) \rvert < 1$ (the region below the red dotted line) is a subset of that determined by (\ref{align:Delta_known_capacity}) (the region below the black dotted line). For this reason, the initial value $\Delta_0$ is already small enough to defend against the DoS attack and achieve consensus. As $\Delta_k \leqslant \Delta_0$, the follow-up values of $\Delta_k$ become inconsequential.

    \begin{example} \label{ex_epsilon0}
        In this example, we justify \textit{Remark \ref{remark:epsilon0}} by considering $(\theta, \ell, \epsilon_0) = (0.9,3,0.2)$. Compared with the sub-graph in the middle of Fig. \ref{fig:stabilization_impulsive_nonlinear_083}, where $\epsilon_0 = 0.01$, the transient response in Fig. \ref{fig:epsilon0} is significantly reduced. In addition, when $\epsilon_0 = 0.2$, the value of $\Delta_k$ is smaller for $t_k \in [2,10.5]$, meaning that more control inputs are required on this time interval as the trade-off.
        \begin{figure}[ht]
        \centering
        \includegraphics[width = 8.3cm]{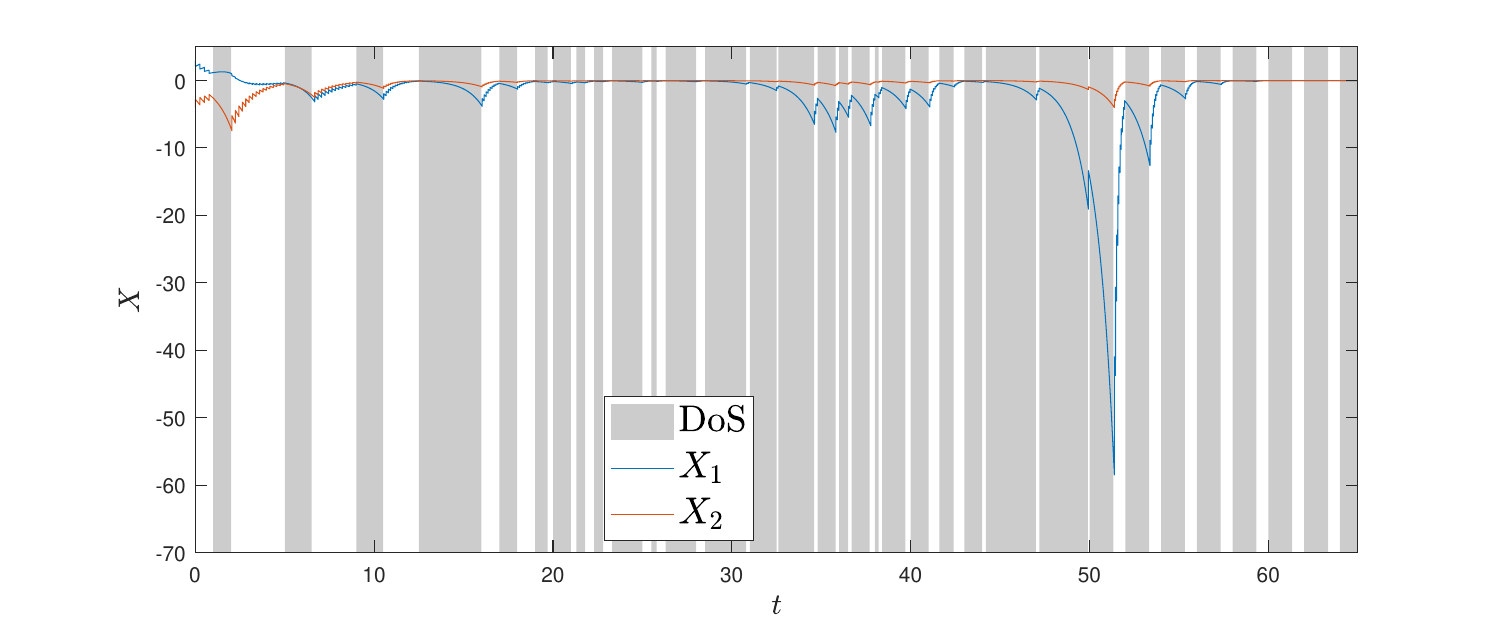}
        \includegraphics[width = 8.3cm]{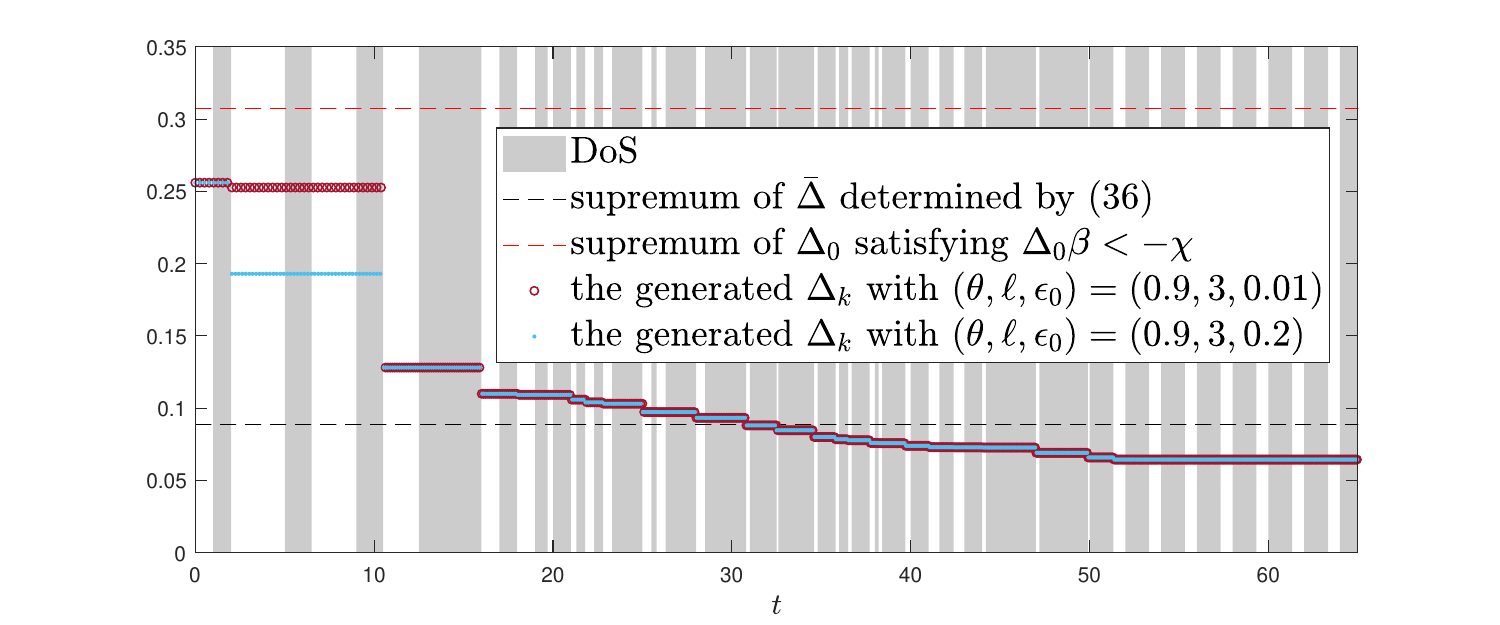}
        \caption{Top: The transient response with $(\theta, \ell, \epsilon_0) = (0.9,3,0.2)$. Bottom: The generated $\Delta_k$ with different values of $\epsilon_0$.}
        \label{fig:epsilon0}
    \end{figure}
    \end{example}

\begin{example} \label{ex3}
    Here, we provide a counterexample to illustrate that $\theta=1$ can lead to the wrong estimation, as discussed in \textit{Remark \ref{remark:theta}}. Specifically, we consider $\xi = \{[2n+1,2n+2)\}_{n \in \mathbb{N}_+}$, and it can be observed that $\inf \mathcal{D} (\xi) = 0.5$ and $ \inf \mathcal{F} (\xi) = 0.5$. If we set $\theta = 1$ in \textit{DoS estimator} (\ref{align:hatT_A(n)})--(\ref{align:hattau_A(n)}), then the estimation $\hat{B}_d (t)$ \& $\hat{B}_f (t)$ is generated as in Fig. \ref{fig:counterexample}. The figure clearly shows the incorrect estimation of $\xi$'s \textit{duration} \& \textit{frequency bounds}.
    \begin{figure}[ht]
        \centering
        \includegraphics[width = 8.3cm]{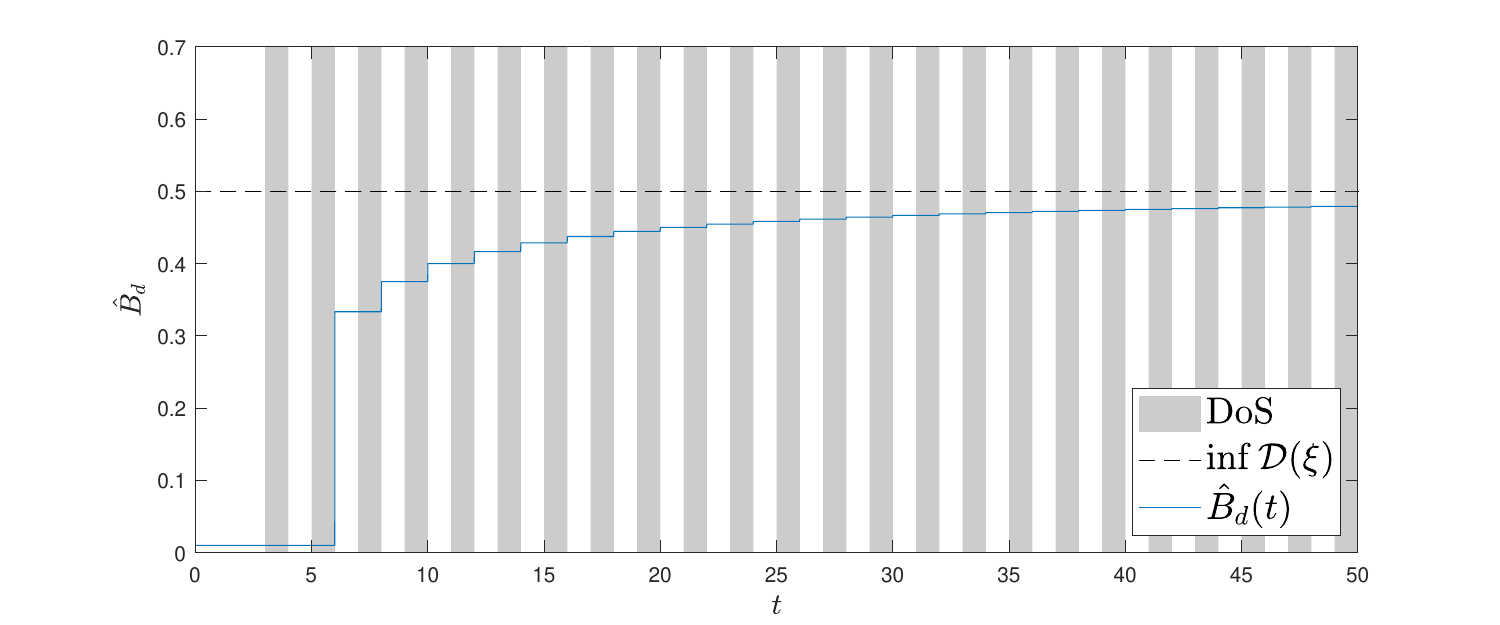}
        \includegraphics[width = 8.3cm]{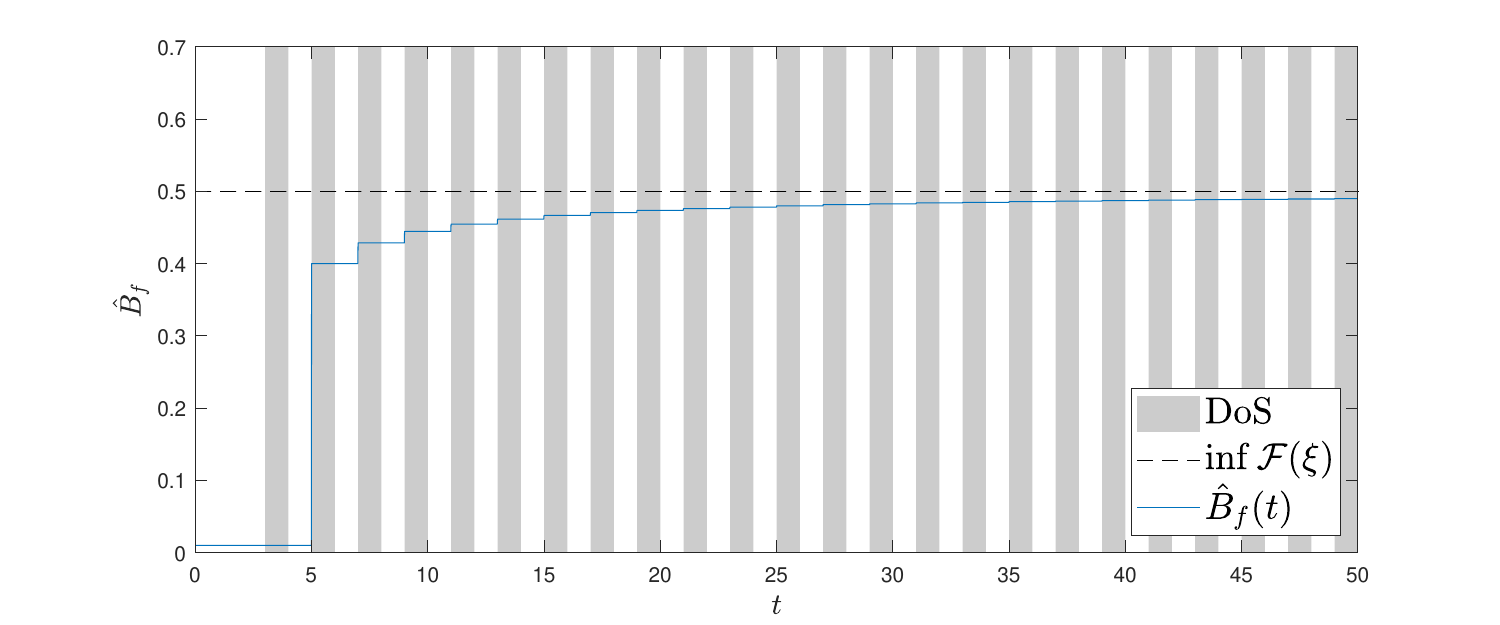}
        \caption{The incorrect estimation of the \textit{duration-bound} and \textit{frequency-bound} by setting $\theta = 1$ and $\ell = 2$.}
        \label{fig:counterexample}
    \end{figure}

    The reason why we fail to estimate the bounds is intuitive: An attacker can hide its strength by behaving as if it is weak, but setting $\theta = 1$ means not having any doubt about the attacker's behavior.
    The necessity of $\theta<1$ suggests that the defender must always be stronger than the attacker in order to defend successfully. Simply matching the attacker's strength could result in failure.
\end{example}

\section{Conclusion}
This paper has addressed the real-time estimation problem of DoS duration and frequency, enabling security control when the attacker is previously unknown to the defender. A new \textit{DoS estimator} has been developed to identify the duration and frequency parameters of any DoS attacker except for three \textit{edge cases}, and it has been successfully applied to two specific control scenarios for illustration.

Future directions include addressing errors in detecting DoS attacks, as accurate detection is essential for the \textit{DoS estimator} (\ref{align:hatT_A(n)})--(\ref{align:hattau_A(n)}). In addition, \textit{Theorem \ref{theorem:upper_bound_of_T}} relies on a priori attacker information to quantify the time required for reliable estimation. It would be valuable to address this issue under milder assumptions.

\appendix
\subsection{Proof of Lemma \ref{lemma:function_limit}} \label{appA}
We first prove (\ref{align:duration_function_limit}). From (\ref{align:Dxi}), we know that, for $\inf \mathcal{D} (\xi) < 1$ and any $\varepsilon \in (0,1- \inf \mathcal{D} (\xi)]$, $\inf \mathcal{D} (\xi) + \varepsilon$ is a \textit{duration-bound} of $\xi$. Thus for any $\varepsilon > 0$, $\exists \, \kappa (\varepsilon) \in  (0,+\infty)$ such that
\begin{align*}
    \lvert \Xi (0,t) \rvert \leqslant \kappa (\varepsilon) + (\inf \mathcal{D} (\xi) + \varepsilon){t}, \; \forall \, t \geqslant 0.
\end{align*}
The inequality above is also true for $\inf \mathcal{D} (\xi) = 1$. Therefore, 
\begin{align*}
\limsup_{t \rightarrow + \infty}{\frac{\lvert \Xi (0,t) \rvert}{t}} \leqslant \inf \mathcal{D} (\xi) + \varepsilon.
\end{align*}
Due to the arbitrariness of $\varepsilon$, we have
\begin{align*}
    \limsup_{t \rightarrow + \infty}{\frac{\lvert \Xi (0,t) \rvert}{t}} \leqslant \inf \mathcal{D} (\xi).
\end{align*}
But if the above inequality is replaced by a strict one, then we can find an $\varepsilon_0 > 0$ that satisfies
\begin{align*}
    \limsup_{t \rightarrow + \infty}{\frac{\lvert \Xi (0,t) \rvert}{t}} < \inf \mathcal{D} (\xi) - \varepsilon_0.
\end{align*}
As a result, there must exist a $T^\ast(\varepsilon_0) \in (0, +\infty)$, such that
\begin{align*}
    \frac{\lvert \Xi (0,t) \rvert}{t} &< \inf \mathcal{D} (\xi) - \frac{\varepsilon_0}{2},\; \forall \, t \geqslant T^\ast.
\end{align*}
Therefore, for any $t\geqslant 0$, we have
\begin{align*}
    \lvert \Xi (0,t) \rvert \leqslant T^\ast + \left(\inf \mathcal{D} (\xi) - \frac{\varepsilon_0}{2}\right)t.
\end{align*}
This, however, means that $\inf \mathcal{D} (\xi) - (\varepsilon_0 /2)$ is a \textit{duration-bound} of $\xi$ and is therefore a contradiction. Thus, (\ref{align:duration_function_limit}) is true.

We can prove (\ref{align:frequency_function_limit}) in a similar way, and the only difference is the case of $\inf \mathcal{F} (\xi) = +\infty$. In such case, we claim that
\begin{align*}
    \limsup_{t \rightarrow + \infty}{\frac{n_{\xi} (0,t)}{t}} = +\infty = \inf \mathcal{F} (\xi).
\end{align*}
Otherwise, we can find $\varepsilon_0 > 0$ and $T^\ast (\varepsilon_0) >0$ such that
\begin{align*}
    \frac{n_{\xi} (0,t)}{t} < \frac{1}{\varepsilon_0}, \; \forall \, t \geqslant T^\ast.
\end{align*}
Thus, we have
\begin{align*}
    n_{\xi} (0,t)
    &\leqslant n_{\xi} (0,T^\ast) + \frac{t}{\varepsilon_0}\\
    &= \frac{n_{\xi} (0,T^\ast)}{T^\ast}T^\ast + \frac{t}{\varepsilon_0}\\
    &< \frac{T^\ast}{\varepsilon_0} + \frac{t}{\varepsilon_0}, \; \forall \, t \geqslant 0,
\end{align*}
which means that $\frac{1}{\varepsilon_0}$ is a \textit{frequency-bound} of $\xi$. This is a contradiction.

If $\inf \mathcal{F} (\xi) < + \infty$, then given any $\varepsilon > 0$, 
$\inf \mathcal{F} (\xi) + \varepsilon$ must be a \textit{frequency-bound} of $\xi$. Due to the arbitrariness of $\varepsilon$, this infers that
\begin{align*}
    \limsup_{t \rightarrow + \infty}{\frac{n_{\xi} (0,t)}{t}} \leqslant \inf \mathcal{F} (\xi).
\end{align*}
But if the above inequality is replaced by a strict one, then we can find an $\varepsilon_0 > 0$, such that
\begin{align*}
    \limsup_{t \rightarrow + \infty}{\frac{n_{\xi} (0,t)}{t}} < \inf \mathcal{F} (\xi) - \varepsilon_0.
\end{align*}
This is a contradiction because it infers that $\inf \mathcal{F} (\xi) - \frac{\varepsilon_0}{2}$ is a \textit{frequency-bound}. Hence, (\ref{align:frequency_function_limit}) is true, and the proof is completed.

\subsection{Proof of Theorem \ref{theorem:bound_reconstruction}} \label{appB}

In view of \textit{Assumption \ref{assumption:not_edge_case}}, $\tau_n$ must have a finite upper bound. For this reason, we have $\inf \mathcal{D} (\xi) = \inf \mathcal{F} (\xi) = 0$ when there are only finite $H_n$s in $\xi$. In that case, the conclusion is trivial with $T = 0$ because $\hat{B}_d (t) \geqslant \hat{B}_d (0) > 0$ and $\hat{B}_f (t) \geqslant \hat{B}_f (0) > 0$.

From now on, we assume that $\xi$ is infinite, and we need the following auxiliary lemma to proceed.
\begin{lemma} \label{lemma:sequence_limit}
    Under \textit{Assumption \ref{assumption:not_edge_case}}, it holds that
    \begin{align} \label{align:duration_sequence_limit}
    \limsup_{n \rightarrow \infty}{\frac{\lvert \Xi (0,h_n+\tau_n) \rvert}{h_n + \tau_n}} = \inf \mathcal{D} (\xi)
    \end{align}
    and
    \begin{align} \label{align:frequency_sequence_limit}
        \limsup_{n \rightarrow \infty}{\frac{n}{h_n}} = \inf \mathcal{F} (\xi).
    \end{align}
\end{lemma}

\begin{proof}
    The proof is based on \textit{Lemma \ref{lemma:function_limit}}. First, we prove (\ref{align:duration_sequence_limit}) by showing that
    \begin{align} \label{align:duration_seq_leq_func}
        \limsup_{n \rightarrow \infty}{\frac{\lvert \Xi (0,h_n+\tau_n) \rvert}{h_n + \tau_n}} \leqslant \limsup_{t \rightarrow + \infty}{\frac{\lvert \Xi (0,t) \rvert}{t}}
    \end{align}
    and
    \begin{align} \label{align:duration_seq_geq_func}
        \limsup_{n \rightarrow \infty}{\frac{\lvert \Xi (0,h_n+\tau_n) \rvert}{h_n + \tau_n}} \geqslant \limsup_{t \rightarrow + \infty}{\frac{\lvert \Xi (0,t) \rvert}{t}}.
    \end{align}
    Recalling \textit{Assumption \ref{assumption:not_edge_case}} and the fact that $\xi$ is infinite, there must hold that $h_n \stackrel{n \rightarrow \infty}{\longrightarrow} +\infty$. Therefore, (\ref{align:duration_seq_leq_func}) is trivial.

    To prove (\ref{align:duration_seq_geq_func}), we will construct a function $\varphi(t)$ such that
    \begin{align*}
        \varphi(t) \geqslant \frac{\lvert \Xi (0,t) \rvert}{t},
    \end{align*}
    and show that
    \begin{align*}
        \limsup_{n \rightarrow \infty}{\frac{\lvert \Xi (0,h_n+\tau_n) \rvert}{h_n+\tau_n}} \geqslant \limsup_{t \rightarrow + \infty}{\varphi(t)}.
    \end{align*}
    Specifically, we define $\varphi (t)$ as
    \begin{eqnarray*}
        \varphi (t) \coloneqq \left\{
        \begin{array}{ll}
            0, & t \in [0,h_1), \\
            \frac{\lvert \Xi (0,h_n+\tau_n) \rvert}{h_n + \tau_n}, & t \in [h_n,h_{n+1}).
        \end{array}
        \right.
    \end{eqnarray*}
    Then, it is clear that
    $$\limsup_{n \rightarrow \infty}{\frac{\lvert \Xi (0,h_n+\tau_n) \rvert}{h_n+\tau_n}} \geqslant \limsup_{t \rightarrow + \infty}{\varphi(t)}.$$
    On the other hand, note that
    $${\frac{\lvert \Xi (0,t) \rvert}{t}} = 0,
    \,\forall \, t \in [0,h_1),$$
    and
    $$\frac{\lvert \Xi (0,h_{n}+\tau_{n}) \rvert}{h_{n}+\tau_{n}} \geqslant {\frac{\lvert \Xi (0,t) \rvert}{t}},
    \,\forall \, t \in [h_{n},h_{n+1}).$$
    Thus, we can conclude that $\varphi(t) \geqslant \frac{\lvert \Xi (0,t) \rvert}{t}$ and (\ref{align:duration_seq_geq_func}) therefore holds. In view of (\ref{align:duration_seq_leq_func}), (\ref{align:duration_seq_geq_func}), and \textit{Lemma \ref{lemma:function_limit}}, then (\ref{align:duration_sequence_limit}) is true.

    The proof of (\ref{align:frequency_sequence_limit}) is quite similar. Since $h_n \rightarrow +\infty$, we clearly have
    \begin{align*}
        \limsup_{n \rightarrow \infty}{\frac{n}{h_n}} = \limsup_{n \rightarrow \infty}{\frac{n_\xi (0,h_n)}{h_n}} \leqslant \limsup_{t \rightarrow +\infty}{\frac{n_\xi (0,t)}{t}}.
    \end{align*}
    On the other hand, if we construct
    \begin{eqnarray*}
        \psi (t) \coloneqq \left\{
        \begin{array}{ll}
            \frac{1}{h_1}, & t \in [0,h_1), \\
            \frac{n}{h_n}, & t \in [h_n,h_{n+1}),
        \end{array}
        \right.
    \end{eqnarray*}
    it is easily derived that 
    \begin{align*}
    \psi (t) \geqslant \frac{n_\xi (0,t)}{t}
    \end{align*}
    and
    \begin{align*}
        \limsup_{t \rightarrow +\infty}{\psi(t)} = \limsup_{n \rightarrow \infty}{\frac{n}{h_n}}.
    \end{align*}
    Consequently, we have by \textit{Lemma \ref{lemma:function_limit}} that
    \begin{align*}
        \limsup_{n \rightarrow \infty}{\frac{n}{h_n}} = \limsup_{t \rightarrow +\infty}\frac{n_\xi (0,t)}{t} = \inf \mathcal{F} (\xi),
    \end{align*}
    and the proof of \textit{Lemma \ref{lemma:sequence_limit}} is completed. \hfill $\square$
\end{proof}

Now, we can continue the proof of \textit{Theorem \ref{theorem:bound_reconstruction}} when $\xi$ is infinite.
Recalling that $\hat{B}_{d} (h_n + \tau_n)$ and $\hat{B}_{f} (h_n)$ take the maximum values from $0$ to $n$, they are clearly non-decreasing in $n$. Hence, it suffices to find a finite $N_1 \geqslant \ell$ that makes $\hat{B}_{d} (h_{N_1} + \tau_{N_1})$ \big(resp., $\hat{B}_{f} (h_{N_1})$\big) the \textit{duration-bound} (resp., \textit{frequency-bound}) of $\xi$.

First, we show the existence of a finite integer $N_2$, such that $\hat{B}_{d} (h_{N_2} + \tau_{N_2})$ is a \textit{duration-bound} of $\xi$. Assume by contradiction, that is, for any $n \geqslant \ell$ and any $\Tilde{\kappa}>0$, there exists a $t^\ast (n, \Tilde{\kappa}) \geqslant 0$ satisfying
\begin{align} \label{T1.1}
    \lvert \Xi (0,t^\ast) \rvert
    &> \Tilde{\kappa} + \hat{B}_{d} (h_n + \tau_n)\,{t^\ast} \nonumber \\
    &\geqslant \Tilde{\kappa} + \left(\max_{\ell \leqslant i \leqslant n}\{\theta B_d (i) + (1-\theta)\} \right) t^\ast\nonumber \\
    &\geqslant \Tilde{\kappa} + \left( \theta B_d (n) + (1-\theta) \right) {t^\ast}.
\end{align}
In view of \textit{Lemma \ref{lemma:sequence_limit}}, there exists a subsequence $\{n_k\}_{k \in \mathbb{N}_+} \subseteq \{n\}_{n \in \mathbb{N}_+}$, such that
\begin{align*}
    \lim_{k \rightarrow \infty}{B_d (n_k)} = \inf \mathcal{D} (\xi).
\end{align*}
Thus, for $\varepsilon_1 \coloneqq \frac{\left(1-\inf \mathcal{D} (\xi)\right) (1-\theta)}{2} > 0$, we can pick an $n_0 (\varepsilon_1) \in \{n_k\}_{k \in \mathbb{N}_+}$ to satisfy
\begin{align*}
    \theta B_d (n_0) + (1 - \theta) &> \theta \inf \mathcal{D} (\xi) + (1-\theta) - \varepsilon_1 \\
    &= \inf \mathcal{D} (\xi) + \varepsilon_1.
\end{align*}
Hence, it follows from (\ref{T1.1}) that, for any $\Tilde{\kappa} > 0$ and $n = n_0$, there exists a $t^\ast(n_0, \Tilde{\kappa}) \geqslant 0$ which satisfies
\begin{align*}
    \lvert \Xi (0,t^\ast) \rvert > \Tilde{\kappa} + \left(  \inf \mathcal{D} (\xi) + \varepsilon_1 \right){t^\ast}.
\end{align*}
This is a contradiction because $\inf \mathcal{D} (\xi) + \varepsilon_1$ is a \textit{duration-bound} of $\xi$. The existence of $N_2$ is therefore proved.

Analogously, we can find an $N_3 \geqslant \ell$ such that $\hat{B}_{f} (h_{N_3})$ is a \textit{frequency-bound} of $\xi$. In fact, by taking $\varepsilon_2 \coloneqq \frac{(1-\theta) \inf \mathcal{F} (\xi)}{2\theta} > 0$, we can use \textit{Lemma \ref{lemma:sequence_limit}} to find an $n_0$ that satisfies
\begin{align*}
\frac{n_0}{\theta h_{n_0}} > \frac{\inf \mathcal{F} (\xi)}{\theta} - \varepsilon_2 = \inf \mathcal{F} (\xi) + \varepsilon_2.
\end{align*}
If the desired $N_3$ does not exist, it will be similarly derived for any $\Tilde{\Lambda} > 0$ that there exists a $t^\ast (n_0,\Tilde{\Lambda}) \geqslant 0$ satisfying
\begin{align*}
    n_\xi (0,t^\ast) > \Tilde{\Lambda} + \left( \inf \mathcal{F} (\xi) + \varepsilon_2 \right){t^\ast}.
\end{align*}
This is a contradiction.

Taking $N_1 = \max\{N_2,N_3\}$ and $T = h_{N_1} + \tau_{N_1}$ (which must be finite because of \textit{Assumption \ref{assumption:not_edge_case}}), the proof of \textit{Theorem \ref{theorem:bound_reconstruction}} is completed.

\subsection{Proof of Theorem \ref{theorem:consensus_MAS_adaptive_sampling}} \label{appC}

The proof will be divided into two parts. In the first part we show that the system (\ref{align:MAS}) achieves consensus if $\Delta_k$ stays within a certain range after some finite instant. In the second part, we further show that $\Delta_k$ designed in (\ref{align:Delta_k_initialize}) and (\ref{align:Delta_k}) converges into the required range within finite time and remains there forever.

The following lemma corresponds to the first part.
\begin{lemma} \label{lemma:Delta_k_condition}
    Consider the sampled data MAS (\ref{align:MAS})--(\ref{align:MAS_ui}) subject to a DoS sequence $\xi$ satisfying \textit{Assumption \ref{assumption:not_edge_case}}. The system achieves consensus if there exist finite constants $\check{t}$, $\underline{\Delta}$, and $\bar{\Delta}$ such that,
    \begin{eqnarray} \label{align:Delta_k_range}
    \left\{
    \begin{array}{lcl}
        0 < \underline{\Delta} \leqslant \Delta_k < +\infty, & \forall \, t_k \geqslant 0, \\
        \Delta_k \leqslant \bar{\Delta}, & \forall \, t_k \geqslant \check{t}.
    \end{array}
    \right.
    \end{eqnarray}
    Here, $\bar{\Delta}$ should satisfy
    \begin{align} \label{align:barDelta1}
        \lvert 1 - \bar{\Delta} \lambda_N (\mathscr{L}) \rvert < 1
    \end{align}
    and
    \begin{align} \label{align:barDelta2}
        B_d + B_f \bar{\Delta} < 1,
    \end{align}
    where $B_d$ and $B_f$ are, respectively, the \textit{duration} \& \textit{frequency bounds} of $\xi$.
\end{lemma}

\begin{proof}
    In order to recast the system (\ref{align:MAS}) into a compact form, we introduce some notations first. We denote by $\bar{x} \coloneqq \frac{1}{N} \sum_{i=1}^{N} x_i$ the mean value of the $N$ agents, and by $e \coloneqq (x_1 - \bar{x}, x_2 - \bar{x}, \cdots, x_N - \bar{x})^T$ the consensus error vector.
    Notably, we have $\bar{x} = \frac{1}{N} \sum_{i=1}^{N} x_{i}(0)$ because $\dot{\bar{x}} = 0$.
    In view of (\ref{align:MAS}) and (\ref{align:MAS_ui}), we can derive that
    \begin{eqnarray*}
    e (t_{k+1}) = 
        \left\{
        \begin{array}{lcl}
            e(t_k), & \text{if } t_k \in \Xi(0,+\infty) \\
            (I - \Delta_k \mathscr{L}) e(t_k), & \text{if } t_k \in \Theta(0,+\infty) 
        \end{array}
        \right.
    \end{eqnarray*}
    It is easy to verify that $e \perp \mathbf{1}_N$, and therefore
    \begin{align} \label{align:matrix_norm}
        &\quad \lvert (I-\Delta_k \mathscr{L}) e(t_k) \rvert \nonumber \\
        &\leqslant \max \left\{\lvert 1 - \Delta_k \lambda_2 (\mathscr{L}) \rvert,\,\lvert 1 - \Delta_k \lambda_N (\mathscr{L}) \rvert \right\} \cdot \lvert e(t_k)\rvert \nonumber \\
        &\leqslant \max \left\{\lvert 1 - \underline{\Delta} \lambda_2 (\mathscr{L}) \rvert,\,\lvert 1 - \bar{\Delta} \lambda_N (\mathscr{L}) \rvert \right\} \cdot \lvert e(t_k)\rvert \nonumber \\
        &\eqqcolon \gamma_2 \lvert e(t_k)\rvert <  \lvert e(t_k)\rvert,\quad \forall \, t_k \geqslant \check{t}.
    \end{align}
    Here, the second inequality follows from (\ref{align:Delta_k_range}), and the inequality $\gamma_2 \lvert e(t_k)\rvert <  \lvert e(t_k)\rvert$ follows from (\ref{align:barDelta1}). If we denote by $t_{k(\check{t})}$ the first transmission attempt no earlier than $\check{t}$, and by $\alpha(\tau,s)$ the number of successful transmissions on $[\tau,s]$, it can be derived from (\ref{align:matrix_norm}) that
    \begin{align*}
        \lvert e(t_m) \rvert &= \gamma_2^{\alpha(\check{t},t_m)-1} \lvert e(t_{k(\check{t})}) \rvert,\,\forall \, m > k(\check{t}).
    \end{align*}
    Since $\check{t}$ is finite and $0< \underline{\Delta} \leqslant \Delta_k < +\infty$, we have finite $t_{k(\check{t})}$ and $ \lvert e(t_{k(\check{t})}) \rvert$. On the other hand, when $\Delta_k \leqslant \bar{\Delta}$ and $B_d + B_f \bar{\Delta} < 1$, it follows from lemma 3 in \cite{DePersis_DoS_nonlinear_scl2016} that $\alpha(\check{t},t_m) \rightarrow +\infty$ as $t_m \rightarrow +\infty$. Together with $\gamma_2 <1$, this yields that $\lvert e(t_m) \rvert \stackrel{t_m \rightarrow +\infty}{\longrightarrow} 0$ and consequently $e(t) \stackrel{t \rightarrow +\infty}{\longrightarrow} 0$.
    Hence, the average consensus is achieved. \hfill $\square$
\end{proof}

Next, we present another lemma to show that $\Delta_k$ designed in (\ref{align:Delta_k_initialize}) and (\ref{align:Delta_k}) converges into the range determined by (\ref{align:Delta_k_range})--(\ref{align:barDelta2}) within finite time.
\begin{lemma} \label{lemma:Delta_k_satisfactory}
    Under \textit{Assumption \ref{assumption:not_edge_case}}, there exists a finite $\check{t} \geqslant 0$, such that $\Delta_k$ generated by (\ref{align:Delta_k_initialize}) and (\ref{align:Delta_k}) satisfies (\ref{align:Delta_k_range})--(\ref{align:barDelta2}).
\end{lemma}

\begin{proof}
    First, it is clear that $\Delta_k \leqslant \Delta_0 < +\infty$, $\forall \, t_k \geqslant 0$. Recalling the definition of $\Delta_0$ below (\ref{align:Delta_k_initialize}), we have (\ref{align:barDelta1}) holds with $\bar{\Delta} = \Delta_0$.
    
    When $\xi$ is finite, $\inf \mathcal{D} (\xi) = \inf \mathcal{F} (\xi) = 0$, and hence any $B_d \in (0,1)$ \& $B_f > 0$ must be the \textit{duration} \& \textit{frequency bounds} of $\xi$. By taking them to be sufficiently small, (\ref{align:barDelta2}) can be satisfied with $\bar{\Delta} = \Delta_0$. On the other hand, because $\xi$ is finite, the set $\{\Delta_k\}_{k \in \mathbb{N}_+}$ only contains finite positive elements and thus has a positive lower-bound $\underline{\Delta}$. Then by taking $\check{t} = 0$ and $\bar{\Delta} = \Delta_0$, (\ref{align:Delta_k_range})--(\ref{align:barDelta2}) are all satisfied.

    When $\xi$ is infinite, it is easy to derive from \textit{Lemma \ref{lemma:sequence_limit}} that
    \begin{align*}
        \lim_{n \rightarrow \infty}{\hat{B}_{d} (h_n + \tau_n)} 
        < 1
    \end{align*}
    and
    \begin{align*}
        \lim_{n \rightarrow \infty}{\hat{B}_{f} (h_n)} 
        < +\infty.
    \end{align*}
    Here, the inequalities are true because $\ell \geqslant 2$ and \textit{Assumption \ref{assumption:not_edge_case}} holds. Since $\hat{B}_{d} (h_n + \tau_n)$ and $\hat{B}_{f} (h_n)$ are non-decreasing in $n$, $\Delta_k$ generated by (\ref{align:Delta_k}) must be non-increasing and lower-bounded by
    \begin{align*}
        \underline{\Delta} &= \min \left\{ {\Delta_0, \lim_{n \rightarrow \infty} \Bigg[\frac{1}{\gamma_1\hat{B}_{f} (h_n)}\left(1- \hat{B}_{d} (h_n + \tau_n)\right)}\Bigg] \right\} \\
        &> 0.
    \end{align*}

    Now it remains to find a $\bar{\Delta}$ satisfying (\ref{align:barDelta1})--(\ref{align:barDelta2}) and a finite $\check{t}$, such that $\Delta_k \leqslant \bar{\Delta}$ holds for any $t_k \geqslant \check{t}$. According to the proof of \textit{Theorem \ref{theorem:bound_reconstruction}}, there exists a finite $T = h_{N_1} + \tau_{N_1}$ satisfying $\hat{B}_d (T) \in \mathcal{D} (\xi)$ and $\hat{B}_f (T) \in \mathcal{F} (\xi)$.
    Then, (\ref{align:Delta_k}) yields for any $t_k \geqslant \check{t} \coloneqq T$ that $$\Delta_k \leqslant \Delta_{k(\check{t})}$$ and
    \begin{align*}
        \hat{B}_d (T) + \Delta_{k(\check{t})} \hat{B}_f (T) &\leqslant \hat{B}_d (T) + \frac{1}{\gamma_1} \left(1- \hat{B}_d (T)\right) \\
        &< 1.
    \end{align*}
    By letting $\bar{\Delta} = \min \{\Delta_0, \Delta_{k(\check{t})}\}$, we have (\ref{align:Delta_k_range})--(\ref{align:barDelta2}) be satisfied with $\check{t} = T$, $B_d = \hat{B}_d (T)$ and $B_f = \hat{B}_f (T)$. The proof is completed. \hfill $\square$
\end{proof}

With \textit{Lemmas \ref{lemma:Delta_k_condition}} and \textit{\ref{lemma:Delta_k_satisfactory}}, the conclusion of \textit{Theorem \ref{theorem:consensus_MAS_adaptive_sampling}} directly follows.

\subsection{Proof of Theorem \ref{theorem:impulsive_stabilization_adaptive_interval}} \label{appD}

According to \textit{Corollary 3.1.1} in \cite{book_impulsive_differential_equation}, if \textit{Assumption \ref{assumption:impulsive_system}} is satisfied, then we have
\begin{align} \label{align:Vt_bound1}
    V(t,X(t)) \leqslant V(0,X_0) \mu^{\alpha (0, t)}  \exp(\beta t).
\end{align}
Here, $\alpha (\tau,s)$ denotes the number of successful control instants on $[\tau,s]$. For further analysis, it is necessary to estimate the value of $\alpha (0, t)$, which is addressed by the following lemma.
\begin{lemma} \label{lemma:number_of_impulses}
    Consider a DoS sequence satisfying \textit{Assumption \ref{assumption:not_edge_case}} and a sequence of \textit{control instants} $\{t_k\}_{k\in \mathbb{N_+}}$ satisfying $\Delta_k \coloneqq t_{k+1} - t_k < +\infty$. If for some finite instant $\check{t} \geqslant 0$, there exists a $\bar{\Delta} > 0$ such that $\Delta_k \leqslant \bar{\Delta},\,\forall \, t_k \geqslant \check{t}$, then, it holds that
    \begin{align} \label{align:number_of_impulses}
        \alpha (0, t) \geqslant \frac{t}{\bar{\Delta}} \left( 1- B_d - \bar{\Delta} B_f \right) - C_1,\,\forall \, t \geqslant t_{k(\check{t})}.
    \end{align}
    Here, $t_{k(\check{t})} < +\infty$ represents the first control instant no earlier than $\check{t}$, $B_d$ and $B_f$ are the \textit{duration} \& \textit{frequency bounds} of $\xi$, $C_1 \coloneqq \frac{t_{k(\check{t})} + \kappa}{\bar{\Delta}} + 1 + \Lambda$, and $\kappa$ and $\Lambda$ are given in \textit{Definitions \ref{def:duration} and \ref{def:frequency}}.
\end{lemma}
\begin{proof}
    Consider any connected component of $\Theta(t_{k(\check{t})},$ $+\infty)$, which must be in the form of $[s_a,s_b)$, $s_b \geqslant s_a$. Here, $\Theta(\tau,s)$ is defined in (\ref{align:Theta}). Since $[s_a,s_b) \subseteq [t_{k(\check{t})},+\infty)$, it contains no less than $\left(({s_b - s_a})/{\bar{\Delta}}\right) - 1$ successful control instants. Therefore, $\alpha (t_{k(\check{t})}, t)$ must be no less than $\lvert \Theta(t_{k(\check{t})},t) \rvert/\bar{\Delta}$ minus the number of connected components of $\Theta(t_{k(\check{t})},t)$. Combining this with $B_d \in \mathcal{D} (\xi)$ and $B_f \in \mathcal{F} (\xi)$, we obtain that 
    \begin{align*}
        &\quad \alpha (t_{k(\check{t})}, t) \\
        &\geqslant \frac{\lvert \Theta(t_{k(\check{t})},t) \rvert}{\bar{\Delta}} - \left(n_\xi (t_{k(\check{t})},t) + 1 \right) \\
        &\geqslant \frac{\lvert \Theta(0,t) \rvert - t_{k(\check{t})}}{\bar{\Delta}} - \left(n_\xi (0,t) + 1 \right) \\
        &\geqslant \frac{1}{\bar{\Delta}} \left( \bigg(1-B_d\bigg)t - t_{k(\check{t})} - \kappa \right) - \left(n_\xi (0,t) + 1 \right) \\
        &\geqslant \frac{t}{\bar{\Delta}}\left(1-B_d-\bar{\Delta}B_f\right) - C_1.
    \end{align*}
    Since $\alpha (0, t) \geqslant \alpha (t_{k(\check{t})}, t)$, (\ref{align:number_of_impulses}) directly follows. \hfill $\square$
\end{proof}

In order to use (\ref{align:number_of_impulses}), we need to show that $\Delta_k$ generated by (\ref{align:impulsive_Delta_k_initialize}) and (\ref{align:impulsive_Delta_k}) satisfy the conditions in \textit{Lemma \ref{lemma:number_of_impulses}}. In addition, the upper-bound $\bar{\Delta}$ in (\ref{align:number_of_impulses}) should satisfy
\begin{align} \label{align:barDelta_impulsive}
    \bar{\Delta}\left( \beta - B_f {\chi} \right) < -\chi \left(1 - B_d\right)
\end{align}
so that $\mu^{\alpha(0,t)} \exp(\beta t) \rightarrow 0$ as $t \rightarrow +\infty$.
To this end, we introduce the following lemma.
\begin{lemma} \label{lemma:Delta_k_satisfactory_impulsive}
    Under \textit{Assumption \ref{assumption:not_edge_case}}, if $1 > B_d \in \mathcal{D}(\xi)$ and $+\infty > B_f \in \mathcal{F}(\xi)$, then $\Delta_k$ generated by (\ref{align:impulsive_Delta_k_initialize}) and (\ref{align:impulsive_Delta_k}) satisfy the following items.
    \begin{enumerate}
        \item $\Delta_k < +\infty$.
        \item There exists a $\underline{\Delta} > 0$, such that $\Delta_k \geqslant \underline{\Delta}$, $\forall \, k \in \mathbb{N}_+$.
        \item There exist a constant $\bar{\Delta} > 0$ satisfying (\ref{align:barDelta_impulsive}) and a finite $\check{t} \geqslant 0$, such that $\Delta_k \leqslant \bar{\Delta},\,\forall \, t_k \geqslant \check{t}$.
    \end{enumerate}
\end{lemma}

\begin{proof}
    The item 1) is trivial. In fact, we have $\Delta_k \leqslant \Delta_0$ according to (\ref{align:impulsive_Delta_k_initialize}) and (\ref{align:impulsive_Delta_k}).

    When $\xi$ is finite, the set $\{\Delta_k\}_{k \in \mathbb{N}_+}$ contains finite positive elements, so $\Delta_k$ must have a positive lower-bound. Hence, the item 2) is satisfied. In addition, we have $\inf \mathcal{D} (\xi) = \inf \mathcal{F} (\xi) = 0$ due to \textit{Assumption \ref{assumption:not_edge_case}}. Then, any finite $B_d \in (0,1)$ and $B_f > 0$ are the \textit{duration} \& \textit{frequency bounds} of $\xi$, respectively. By choosing them to be sufficiently small, (\ref{align:barDelta_impulsive}) is satisfied with $\bar{\Delta} = \Delta_0$. Thus, the item 3) is satisfied with $\bar{\Delta} = \Delta_0$ and $\check{t} = 0$.

    When $\xi$ is infinite,  it is easy to derive from \textit{Lemma \ref{lemma:sequence_limit}} that
    \begin{align*}
        {\lim_{n \rightarrow \infty}{\hat{B}_{d} (h_n + \tau_n)} 
        < 1}
    \end{align*}
    and
    \begin{align*}
        {\lim_{n \rightarrow \infty}{\hat{B}_{f} (h_n)} 
        < +\infty.}
    \end{align*}
    Recalling the monotonicity of $\hat{B}_{d} (h_n + \tau_n)$ and $\hat{B}_{f} (h_n)$, we have $\Delta_k$ generated by (\ref{align:impulsive_Delta_k}) non-increasing. Consequently, $\Delta_k$ is lower-bounded by
    \begin{align*}
        \underline{\Delta} = \lim_{n \rightarrow \infty} \frac{\chi \left(1-\hat{B}_{d} (h_n + \tau_n)\right)}{\gamma_3 \left(\hat{B}_{f} (h_n){\chi} - \beta \right)} > 0.
    \end{align*}
    Hence, the item 2) is satisfied. As for 3), it is also satisfied if we take $\check{t} \coloneqq T$ and $\bar{\Delta} = \Delta_{k(\check{t})}$ where the finite $T$ is determined by \textit{Theorem \ref{theorem:bound_reconstruction}}. In fact, it follows from (\ref{align:impulsive_Delta_k}) that
    \begin{align*}
        &\qquad \bar{\Delta} \left( \beta - \hat{B}_f (T) {\chi} \right) \\
        &\leqslant \frac{\chi \left(1-\hat{B}_d (T)\right)}{\gamma_3 \left(\hat{B}_f (T) {\chi} - \beta \right)} \left( \beta - \hat{B}_f (T) {\chi} \right) \\
        &< - \chi \left(1-\hat{B}_d (T)\right),
    \end{align*}
    where the second inequality is because $\gamma_3 > 1$, $\beta > 0$ and $\chi = \ln \mu < 0$. Hence, $\bar{\Delta}$ satisfies (\ref{align:barDelta_impulsive}) as $\hat{B}_d (T)$ \& $\hat{B}_f (T)$ are the \textit{duration} \& \textit{frequency bounds} of $\xi$. Additionally, the monotonicity of $\Delta_k$ indicates that $\Delta_k \leqslant \bar{\Delta}$ for any $t_k \geqslant \check{t}$. The proof is completed. \hfill $\square$

    With \textit{Lemmas \ref{lemma:number_of_impulses}} and \textit{\ref{lemma:Delta_k_satisfactory_impulsive}}, we can substitute (\ref{align:number_of_impulses}) into (\ref{align:Vt_bound1}) and obtain the following inequality
    \begin{align*}
        V(t,X(t)) &\leqslant V(0,X_0)  e^{\beta t} \mu^{\frac{t}{\bar{\Delta}} \left( 1- B_d - \bar{\Delta} B_f\right) - C_1} \\
        &\leqslant V(0,X_0) \mu^{-C_1} e^{\beta t  + \chi \left( \frac{t}{\bar{\Delta}} \left( 1- B_d - {\bar{\Delta}}B_f\right) \right)} \\
        &= V(0,X_0) \mu^{-C_1} e^{\left(\bar{\Delta} \left(\beta - B_f{\chi}\right) + \chi \left( 1- B_d\right) \right) \frac{t}{\bar{\Delta}}}, \\
        & \qquad \qquad \qquad \qquad \qquad \forall \, t \geqslant t_{k(\check{t})}.
    \end{align*}
    Since $\Delta_k < +\infty$, $t_{k(\check{t})}$ must be finite. And due to 3) in \textit{Assumption \ref{assumption:impulsive_system}}, $V(t,X(t))$ cannot escape on $[0,t_{k(\check{t})}]$. In view of (\ref{align:barDelta_impulsive}), $V(t,X(t))$ exponentially converges to $0$, and thus the system (\ref{align:nonlinear_system}) is exponentially stabilized by the impulsive controller (\ref{align:impulsive_control})--(\ref{align:impulsive_invalid_control}) (see Theorem 3.2.1 in \cite{book_impulsive_differential_equation}). In addition, $\Delta_k$ is lower-bounded by $\underline{\Delta} > 0$ according to \textit{Lemma \ref{lemma:Delta_k_satisfactory_impulsive}}.
\end{proof}


\bibliographystyle{ieeetr}
\bibliography{References_of_DoS_Estimation}
\end{document}